\documentclass[11pt]{article}

\usepackage{amsmath,amssymb,amsfonts,amsthm,mathtools}
\usepackage{thmtools}
\makeatletter
\@ifundefined{newcounteralias}{}{%
  \renewcommand\thmt@autorefsetup{%
    \@xa\def\csname\thmt@envname autorefname\@xa\endcsname
      \@xa{\thmt@thmname}%
  }%
}
\IfFormatAtLeastTF{2024-11-01}{
  \renewcommand*\@addtoreset[2]{%
    \bgroup
      \edef\aliasctr@@truelist{\aliasctr@follow{#2}}%
      \let\@elt\relax
      \expandafter\@cons\aliasctr@@truelist{{#1}}%
    \egroup
    \expandafter\xdef\csname theH#1\endcsname{%
      \expandafter\noexpand\csname theH#2\endcsname.%
      \noexpand\the\noexpand\value{#1}}%
  }
}{}
\makeatother
\usepackage{thm-restate}

\usepackage[utf8]{inputenc}
\usepackage[T1]{fontenc}
\usepackage[tt=false, type1=true]{libertine}
\usepackage[varqu]{zi4}
\usepackage[libertine]{newtxmath}
\usepackage[margin=1in]{geometry}
\usepackage{microtype,parskip}
\usepackage{placeins}
\usepackage{float}

\usepackage{graphicx,xcolor}
\usepackage{booktabs,array,multirow,tabularx}

\usepackage[ruled,vlined]{algorithm2e}
\SetKwInput{KwRequire}{Requires}
\SetAlgoNoLine
\DontPrintSemicolon
\LinesNumbered
\SetKwComment{Comment}{\textcolor{blue}{// }}{}
\newcommand{\cmt}[1]{\Comment*[r]{\textcolor{blue}{#1}}}

\usepackage{subcaption}
\definecolor{ForestGreen}{rgb}{0.1333,0.5451,0.1333}
\definecolor{DarkRed}{rgb}{0.65,0,0}
\usepackage{hyperref}
\hypersetup{
    colorlinks=true,
    linkcolor=DarkRed,
    filecolor=magenta,
    urlcolor=cyan,
    citecolor=ForestGreen,
    linktocpage=true,
    pagebackref=true,
    bookmarks=true,
    bookmarksopen=true,
    bookmarksnumbered=true
}
\usepackage[nameinlink]{cleveref}
\usepackage{csquotes}
\usepackage[style=alphabetic,natbib=true,maxalphanames=3,minalphanames=3,maxbibnames=99,maxcitenames=3]{biblatex}
\usepackage{tikz}
\usetikzlibrary{arrows.meta,positioning,calc,decorations.pathreplacing,matrix,patterns,decorations.pathmorphing,decorations.markings,fit}
\tikzset{
  paper graph vertex/.style={circle,draw=black!55,fill=white,
    line width=0.3pt,minimum size=4.1mm,inner sep=0pt,
    font=\fontsize{8.5}{10}\selectfont},
  paper graph edge/.style={line width=0.7pt,line cap=round,line join=round},
  paper graph arc/.style={paper graph edge,
    -{Latex[length=1.75mm,width=1.15mm]}},
  paper forest edge/.style={paper graph edge,line width=1.05pt},
  paper forest arc/.style={paper graph arc,paper forest edge},
  paper highlighted edge/.style={paper forest edge,line width=1.4pt},
  paper highlighted arc/.style={paper graph arc,paper highlighted edge,
    -{Latex[length=2mm,width=1.35mm]}},
  paper pending edge/.style={paper graph edge,draw=black!50,
    dash pattern=on 2.5pt off 1.7pt},
  paper step/.style={paper graph edge,
    -{Latex[length=2mm,width=1.35mm]}}
}

\usepackage{pdfcomment}

\usepackage[showdeletions]{color-edits}

\usepackage{xspace,nicefrac}

\newcommand{\flowedge}{arc}

\theoremstyle{plain}
\newtheorem{theorem}{Theorem}
\newtheorem{corollary}{Corollary}[section]
\newtheorem{lemma}[corollary]{Lemma}

\newtheorem{proposition}[corollary]{Proposition}
\newtheorem{definition}[corollary]{Definition}

\AddToHook{cmd/appendix/before}{\crefalias{section}{appendix} \crefalias{subsection}{appendix}}

\usepackage{needspace,xurl}
\hypersetup{pdftitle={Constructive Arboricity in Almost-Linear Time for Dense Graphs},pdfauthor={Iman Gholami, MohammadTaghi Hajiaghayi, Mahdi JafariRaviz, Alireza Kaviani, Soheil Mohammadkhani, AmirMohammad Shahrezaei}}
\newcommand{\dotcup}{\mathbin{\dot{\cup}}}

\newcommand{\AdditiveOne}{\textsc{AdditiveOne}}
\newcommand{\StarInsertion}{\textsc{StarInsertion}}
\newcommand{\Arboricity}{\textsc{Arboricity}}

\SetKwInput{KwIn}{Input}
\SetKwInput{KwOut}{Output}
\SetKw{Return}{return}
\SetKwFunction{Peel}{Peel}
\SetKwFunction{ExactKCall}{ExactK}
\SetKwFunction{AdditiveOneCall}{AdditiveOne}
\crefname{algocf}{algorithm}{algorithms}
\Crefname{algocf}{Algorithm}{Algorithms}
\crefname{observation}{observation}{observations}
\Crefname{observation}{Observation}{Observation}
\title{Constructive Arboricity in Almost-Linear Time for Dense Graphs}
\author{
Iman Gholami\textsuperscript{1}\qquad
Mohammad T. Hajiaghayi\textsuperscript{1}\qquad
Mahdi JafariRaviz\textsuperscript{1}\\
Alireza Kaviani\textsuperscript{2}\qquad
Soheil Mohammadkhani\textsuperscript{2}\qquad
AmirMohammad Shahrezaei\textsuperscript{2}\\[1.5ex]
{\normalsize\textsuperscript{1}University of Maryland\qquad
\textsuperscript{2}Independent Researcher}
}
\date{}

\begin{document}
\maketitle
\begingroup
\renewcommand{\thefootnote}{}
\begin{NoHyper}
\footnotetext{\raggedright Email: \texttt{\{imangholami77,\allowbreak hajiaghayi,\allowbreak mahdijafariraviz,\allowbreak akaviani05,\allowbreak mohammadkhani.soheil,\allowbreak s.a.m.shahrezaei\}@gmail.com}}
\end{NoHyper}
\endgroup

\begin{abstract}
The arboricity of a graph is the minimum number of forests needed
to partition its edges. We present an algorithm that
computes the exact arboricity of a simple graph and constructs an optimal forest partition in $\min\{n^{2+o(1)},m^{4/3+o(1)}\}$
time, where $n$ and $m$ denote the numbers of vertices and edges.
This improves the previous $O(m^{3/2}\log(n^2/m))$ bound of Gabow.
In particular, it runs in almost-linear time, $m^{1+o(1)}$, on dense
graphs with $m=\Theta(n^2)$.

We also give an algorithm that constructs a partition into at most
one more forest than optimal in $m^{1+o(1)}$ time. To our knowledge,
it is the first almost-linear time algorithm for this problem.
\end{abstract}

\newpage
\tableofcontents
\newpage


\section{Introduction}
\label{sec:intro}

Let $G=(V,E)$ be a simple undirected graph with $n=|V|$ vertices and $m=|E|$ edges. We write $G[X]=(X,E_G(X))$ for the induced subgraph on a vertex set $X\subseteq V$. The \emph{arboricity} of $G$, denoted by $a(G)$, is the minimum number of forests needed to partition its edges. The Nash--Williams characterization \cite{NashWilliams} gives
\begin{equation}
    \label{eq:arboricity}
    a(G) = \max_{\substack{X\subseteq V\\|X|\geq 2}} \left\lceil \frac{|E_G(X)|}{|X|-1} \right\rceil.
\end{equation}
The arboricity is directly related to the \emph{pseudoarboricity} $\rho(G)$ of the graph, defined as the minimum number of pseudoforests needed to partition its edges. A pseudoforest is a graph in which each connected component contains at most one cycle. The corresponding density characterization of pseudoarboricity is as follows:
\begin{equation}
    \rho(G) = \max_{\substack{X\subseteq V\\|X|\geq 2}} \left\lceil \frac{|E_G(X)|}{|X|} \right\rceil.
\end{equation}
For simple graphs, $\rho(G) \le a(G) \le \rho(G) + 1$ (see \Cref{lem:prelim-density-arboricity}). A $k$-out orientation of a graph is an orientation of its edges such that each vertex has outdegree at most $k$. Such an orientation is equivalent to a partition into $k$ pseudoforests: assign distinct colors to the outgoing edges at each vertex, or orient each pseudoforest with outdegree at most one.

Arboricity is a fundamental measure of graph sparsity. Many algorithms are faster on graphs of bounded arboricity, including algorithms for listing triangles and other small subgraphs~\cite{chiba1985arboricity,eppstein1994arboricity} and distributed algorithms for symmetry breaking~\cite{barenboim2010sublogarithmic}. We study how fast one can compute $a(G)$ together with a partition of $E$ into $a(G)$ forests.

Computing the pseudoarboricity is easier. Using almost-linear exact maximum flow~\cite{van2023deterministic}, we can compute $\rho(G)$ and a corresponding orientation in $m^{1+o(1)}$ time, since testing a candidate outdegree bound requires one maximum-flow computation, and binary search finds the minimum bound using $O(\log m)$ such computations; see \Cref{prop:prelim-minimum-orientation}. This also constructs a pseudoforest partition. This flow approach goes back to \citet{picard1982network}, and \citet{goldberg1984finding} used flows in the same way to find a subgraph of maximum density. Faster approximations are also known. Repeatedly removing a vertex of minimum degree gives an acyclic orientation, and hence at most $2a(G)-1$ forests, in linear time~\cite{eppstein1994arboricity,arikati1997efficient}. \citet{kowalik2006approximation} computes a $\lceil(1+\varepsilon)\rho(G)\rceil$-out orientation in $O(\varepsilon^{-1}m\log n\log\rho(G))$ time. \citet{BlumenstockFischer} convert $k$ pseudoforests into $k+1$ forests in $O(mk\log k+m\log n)$ time, giving $\rho(G)+1$ forests when supplied with an optimal pseudoforest partition. This improves the $O(mn\log k)$ conversion from the work of Gabow and Westermann~\cite{westermann1988efficient,gabow1992forests}. Combined with the scheme of Kowalik, this constructs at most $\lceil(1+\varepsilon)a(G)\rceil+1$ forests in $O(\varepsilon^{-1}m\log n\log a(G))$ time~\cite{BlumenstockFischer}.

Constructing an optimal forest partition has proved much harder. The forests of $G$ are the independent sets of a matroid, so the matroid partition algorithm of \citet{Edmonds} computes $a(G)$ and an optimal partition in polynomial time. \citet{picard1982network} reduce the problem to $O(n)$ maximum-flow computations, and \citet{gabow1992forests} specialize matroid partition to forests. The fastest previous algorithm is due to \citet{gabow1998graphic} and runs in $O(m^{3/2}\log(n^2/m))$ time. A second route goes through arborescences. Orient $G$ so that every vertex has outdegree at most $k$, add a root vertex, and add arcs from each vertex to the root until every vertex has outdegree exactly $k$. In the resulting directed multigraph, a partition of $E(G)$ into $k$ forests corresponds to $k$ arc-disjoint arborescences (\Cref{lem:small}). This reduction has been used before~\cite{berczi2010packing,arkhipov2026faster}. Edmonds characterized when such arborescences exist~\cite{edmonds1973edge}, and Gabow's algorithm~\cite{gabow1991matroid} finds them. Combined with exact orientation by flow, Gabow's algorithm gives a constructive bound of $m^{1+o(1)}+O(m\cdot a(G)\log n)$. This is at most $m^{3/2+o(1)}$ for simple graphs, since $|E_G(X)|\leq\min\{m,\binom{|X|}{2}\}$ and \Cref{eq:arboricity} gives $a(G)=O(\sqrt m)$. The bound is also almost-linear when $a(G)$ is constant.

\citet{arkhipov2026faster} study the more general $k$-forest problem: find $k$ forests with the maximum total number of edges. With $k=a(G)$, an optimal solution covers all edges, so these algorithms also construct an optimal forest partition. Matroid partition algorithms specialized to forests solve this problem in $O(m+(kn)^{3/2})$ time up to polylogarithmic factors~\cite{blikstad2023fast,quanrud2024faster}, with other trade-offs between $k$, $n$, and $m$ given by \citet{gabow1992forests}. \citet{arkhipov2026faster} give an $O(k^3\min\{kn,m\}\log^2 n+km^{1+o(1)})$ bound, which is almost linear when $k$ is constant.

Other methods compute only the value $a(G)$. Recently, \citet{cen2025fast} reduced computing $a(G)$ to directed global minimum cut, even in weighted graphs. This reduction does not construct the forests. Concurrently with and independently of our work, \citet{fleischmann2026directedglobalminimumcut} gave an almost-linear-time algorithm for directed global minimum cut. Combined with the reduction of \citet{cen2025fast}, it computes the value $a(G)$, but not the forests, in almost-linear time.

In this work, we present an algorithm that computes the exact arboricity of a simple graph and constructs an optimal forest partition in $\min\{n^{2+o(1)}, m^{4/3+o(1)}\}$ time.

\subsection{Main results}

Our main result is
an algorithm that computes the exact arboricity of a simple graph and
constructs an optimal forest partition in
$\min\{n^{2+o(1)},m^{4/3+o(1)}\}$ time.
This bound improves on the previous
$O\!\left(m^{3/2}\log\frac{n^2}{m}\right)$ time bound of Gabow for
constructing an optimal forest partition~\cite{gabow1998graphic}.
For dense graphs with $m=\Theta(n^2)$, the $n^{2+o(1)}$ bound yields
an almost-linear time construction.

\begin{theorem}[Constructive exact arboricity]
\label{thm:main}
For every simple unweighted undirected graph $G$ with $n$ vertices
and $m$ edges, an algorithm returns $a(G)$
and a partition of $E(G)$ into $a(G)$ forests.
It runs in
\begin{equation*}
  \min\{n^{2+o(1)},m^{4/3+o(1)}\}
\end{equation*}
time.
\end{theorem}

We begin by computing a
$k$-out orientation of $G$ with the smallest possible $k$ in
$m^{1+o(1)}$ time (\Cref{prop:prelim-minimum-orientation}). This minimum
$k$ is exactly the pseudoarboricity $\rho(G)$
(\Cref{lem:prelim-orientation-pseudoforest}). Since the arboricity
$a(G)$ is either $\rho(G)$ or $\rho(G)+1$, we have
$k\leq a(G)\leq k+1$ (\Cref{lem:prelim-density-arboricity}).

The full algorithm then repeatedly removes vertices of current
degree at most $k$, leaving an induced subgraph that we call the
\emph{core}; if nonempty, it has minimum degree at least $k+1$
(\Cref{alg:arboricity,sec:exactpeeling}). Any partition of the core into
at least $k$ forests extends to $G$ by restoring the removed vertices
in reverse order. Each restored vertex has at most $k$ edges to vertices
already present, so assigning these edges to distinct forests creates
no cycle (\Cref{lem:low-degree-insertion,lem:peeling}).
Thus it suffices to construct a
$k$-forest partition of the core if one exists, and otherwise construct
a $(k+1)$-forest partition, before restoring the peeled vertices.

Assume the core has $N$ vertices and $M$ edges. We show that $M=\Theta(Nk)$ after peeling (\Cref{lem:core-size}). This relation is the key to our running-time analysis. We give an additive-one algorithm that converts any supplied
$k$-out orientation into a $(k+1)$-forest partition (\Cref{thm:additive}).
Applying it to the orientation restricted to the core takes
$O(M\log^2 N)$ time, since $M=\Theta(Nk)$. We save this
$(k+1)$-forest partition as a fallback and then try to construct a
$k$-forest partition of the core. If no such partition exists, then
$a(G)=k+1$, and we return the saved partition after restoring the
peeled vertices.

For this remaining task, we use two algorithms, each of which either
constructs a $k$-forest partition or correctly reports that none
exists. Our first algorithm, \textsc{DeleteAndReinsert}, runs on the
core in $M^{1+o(1)}\cdot\frac{N}{k}$ time
(\Cref{sec:deletion,thm:deletion-reinsertion}). The second reduces the
problem to packing $k$ spanning trees after adding an artificial root
and applies Gabow's round-robin algorithm~\cite{gabow1991matroid}. We call it the \emph{rooted construction}
(\Cref{sec:small,prop:gabow-round-robin}).
The resulting running time is $O(nk^2\log n)$ on an $n$-vertex graph;
on the core, it becomes $O(Nk^2\log N)=O(kM\log N)$
(\Cref{thm:small}).

We show that \textsc{DeleteAndReinsert} alone gives a total running time
of $n^{2+o(1)}$. To obtain the bound
in terms of $m$, we choose between the two algorithms by comparing the
factors $k$ and $\frac{N}{k}$: we use the rooted construction when
$k^2\leq N$ and \textsc{DeleteAndReinsert} otherwise. Up to
subpolynomial factors, this choice runs the faster of the two, so it
suffices to bound one of them in each of the following cases.
If $k\leq m^{1/3}$, the rooted construction
takes $O(m^{4/3}\log n)$ time. If $k>m^{1/3}$, the relation
$N=\Theta\!\left(\frac{M}{k}\right)$ on the core bounds the running time of \textsc{DeleteAndReinsert} by
$\frac{M^{2+o(1)}}{k^2}$, which is at most $m^{4/3+o(1)}$.
Combining the two estimates yields the claimed
$\min\{n^{2+o(1)},m^{4/3+o(1)}\}$ running time
(\Cref{sec:full-algorithm}).

\paragraph{\textsc{DeleteAndReinsert} Algorithm}
A central ingredient of our result is the
\textsc{DeleteAndReinsert} algorithm. Given a graph $G$ with $n$ vertices,
$m$ edges, and a $k$-out orientation, it decides whether
$E(G)$ can be partitioned into $k$ forests and constructs such a
partition whenever one exists (\Cref{alg:delete-reinsert} and \Cref{thm:deletion-reinsertion}).

\begin{restatable}[Exact partition by \textsc{DeleteAndReinsert}]{theorem}{thmDeleteAndReinsert}
\label{thm:deletion-reinsertion}
Given a graph $G$ with $n$ vertices, $m$ edges,
and a supplied $k$-out orientation,
\Cref{alg:delete-reinsert} constructs a $k$-forest partition or
reports that no such partition exists in
\begin{equation*}
  (m+nk)^{1+o(1)}\cdot\frac{n}{k}
\end{equation*}
time.
\end{restatable}

The first step finds a small vertex set $R$ whose deletion
leaves a $(k-1)$-out orientation (\Cref{alg:deletion-set}). It suffices
to ensure that every vertex of outdegree $k$ in the original orientation
either belongs to $R$ or has an out-neighbor in $R$. Since each vertex
of $R$ must later be reinserted, keeping $R$ small is essential to the
running time.

We construct $R$ greedily. While a vertex of outdegree $k$ remains,
we delete a vertex whose removal makes as many such vertices as
possible cease to have outdegree $k$. We show that each deletion
makes at least a $\frac{k+1}{n}$ fraction of the remaining vertices of outdegree $k$ no longer have outdegree $k$. Thus the
process terminates after $O\!\left(\frac{n}{k}\log n\right)$ steps, yielding a valid set
$R$ of size $O\!\left(\frac{n}{k}\log n\right)$. We implement the greedy rule efficiently with a maximum heap, finding $R$ in $O((m+n)\log n)$ time (\Cref{lem:deletion}).

After deleting $R$, the restricted orientation of $G-R$ is
$(k-1)$-out. Here we again use our additive-one algorithm, which
converts an $x$-out orientation into an $(x+1)$-forest partition
(\Cref{thm:additive}). Applying it with $x=k-1$ constructs a partition
$\mathcal F=(F_1,\ldots,F_k)$ of $E(G-R)$ into $k$ forests in
$O(m\log n+nk\log^2 n)$ time.

We then reinsert the vertices of $R$ one at a time. Reinserting a vertex
$r$ adds all its edges to vertices already present, which form a star.
We introduce a procedure called \textsc{StarInsertion} that either
extends $\mathcal F$ to a $k$-forest partition of the enlarged graph,
possibly changing the forests assigned to earlier edges, or correctly
reports that no such partition exists. Each call takes
$(m+nk)^{1+o(1)}$ time (\Cref{sec:star-insertion,thm:star}). A failed
insertion proves $a(G)>k$. If every insertion succeeds, the resulting
forests partition $E(G)$. Finding $R$, applying the additive-one
algorithm to $G-R$, and making
$|R|=O\!\left(\frac{n}{k}\log n\right)$ calls to \textsc{StarInsertion}
together take $(m+nk)^{1+o(1)}\cdot\frac{n}{k}$ time,
with the logarithmic factor absorbed into the subpolynomial term
(\Cref{thm:deletion-reinsertion}).

\paragraph{\textsc{AdditiveOne} Algorithm}
Blumenstock and Fischer~\cite{BlumenstockFischer} convert $k$
pseudoforests into $k+1$ forests in $O(mk\log k+m\log n)$ time. Our algorithm builds on their conversion and improves this bound by
choosing a pseudoforest partition with few cycles.

\begin{restatable}[Constructive additive-one decomposition]{theorem}{thmAdditive}
\label{thm:additive}
Given a simple graph $G$ with $n$ vertices and $m$ edges, and a $k$-out orientation, there is an algorithm that constructs $k+1$ forests partitioning $E(G)$ in $O(m\log n+nk\log^2 n)$ time.
\end{restatable}

To find pseudoforests with few cycles, we build one pseudoforest at a time, with each vertex of maximum outdegree selecting one outgoing edge. Each vertex greedily selects an edge that does not create a cycle. A vertex is forced to close a cycle only when all of its outgoing edges stay inside its current component, formed by the edges selected so far. Since the graph is simple, this component then has more vertices than the current maximum outdegree $d$. Each component of a pseudoforest has at most one cycle, so the pseudoforest has at most $n/(d+1)$ cycles. Removing its edges lowers the maximum outdegree by one, and summing over $d=k,k-1,\ldots,1$ gives $O(n\log k)$ cycles in total. To make the greedy choices fast, we process the components in phases, in each of which the number of vertices still needing a choice at least halves. \Cref{sec:lowcycles} gives the full procedure and shows that it runs in $O(m+nk\log n)$ time.

We then use the conversion of \citet{BlumenstockFischer}. We show that its running time is $O(m\log n+nk+Zk\log n)$, where $Z$ is the total number of cycles in the pseudoforests. Our construction gives $Z=O(n\log k)$, so the conversion takes $O(m\log n+nk\log k\log n)$ time. As in the full algorithm, we peel the graph to its core, where $M=\Theta(Nk)$, so the conversion time bound on the core becomes $O(M\log^2 N)$. Restoring the peeled vertices then gives $k+1$ forests of $G$ in $O(m\log^2 n)$ total time. Taking $k=\rho(G)$, this gives an almost-linear time construction with at most $a(G)+1$ forests, since $\rho(G)\leq a(G)$ (\Cref{cor:near-linear-additive-one}).
\paragraph{\textsc{StarInsertion} Algorithm}
A key novelty of our algorithm is the \textsc{StarInsertion} method, which allows us to remove a vertex, solve the smaller instance, and efficiently reinsert the removed vertex along with all its incident edges.
 
\begin{restatable}[Star Insertion]{theorem}{thmStarInsertion}
\label{thm:star}
Let $H$ be a simple graph with $n=|V(H)|$, $m=|E(H)|$, and a specified
vertex $r$. Suppose $E(H-r)$ is partitioned into $k$ forests
$F_1,\ldots,F_k$. Given this partition, an algorithm
either returns a $k$-forest partition of $E(H)$, possibly recoloring
edges of $H-r$, or correctly reports that no such partition exists.
It runs in $(nk+m)^{1+o(1)}$ time.
\end{restatable}

Here, we assume that all edges except a star centered at vertex $r$ are partitioned into $k$ forests, where each forest corresponds to a color.
The first idea is to orient each forest so that maintaining a forest partition reduces to keeping outdegree at most one and avoiding directed cycles.
In each forest, we root every component at an arbitrary vertex and direct all edges toward the roots, so every vertex has outdegree at most one in each color (\Cref{def:star-rooted}); our method ensures that $r$ remains a root in all colors after insertion.
Given this orientation, consider as a simple case inserting a single edge $rv$ when all other edges are colored.
If the endpoint $v$ is \emph{free} of some color (i.e., has no outgoing edge of that color), we can simply assign $rv$ that color with direction $v \to r$.
Otherwise, we can pick a color $c$, assign $rv$ color $c$ with direction $v \to r$, and remove the color and direction of the outgoing edge of color $c$ from $v$.
We call this operation a \emph{color-shift} (\Cref{def:star-color-shift}).
This leaves a new uncolored edge to insert; if its other endpoint is free of some color, we can again finish by using that color.
Thus, when a $k$-forest partition exists, a sequence of color-shifts can reach a vertex with a free color where we can easily insert the remaining uncolored edge (\Cref{def:star-shift-sequence,lem:star-sequence-exists}).

Color-shifts maintain outdegree at most one in each color (\Cref{obs:star-shift-orientation,obs:star-shift-sequence}), but an arbitrary sequence may create directed cycles, which would invalidate the forest partition.
To insert a single edge without creating cycles, we compute a sequence of color-shifts using Dijkstra's algorithm on a directed weighted auxiliary graph whose color-change costs depend on tree depths in the input forests (\Cref{def:star-auxiliary-graph,lem:one-edge-dijkstra}).
Similarly, to insert multiple edges incident to $r$, we add capacities to the auxiliary graph and use minimum-cost maximum flow to obtain edge-disjoint sequences of color-shifts (\Cref{def:star-flow-network,lem:star-paths-orientation}).
In both Dijkstra's algorithm and minimum-cost maximum flow, these depth-based costs in the auxiliary graph ensure that any directed cycle in the resulting color partition contradicts minimum-cost optimality (\Cref{lem:one-edge-dijkstra,lem:star-forest}).
If the star cannot be fully inserted, a minimum-cut argument shows that no $k$-forest partition exists (\Cref{lem:star-failure}).
This auxiliary graph has $O(nk)$ vertices and edges; by running one exact minimum-cost maximum flow computation~\cite{van2023deterministic}, we obtain the running time in \Cref{thm:star}.


\section{Preliminaries}
\label{sec:preliminaries}

Throughout the paper, $G=(V,E)$ is a simple undirected graph with
$n=|V|$ vertices and $m=|E|$ edges. We denote the degree of a
vertex $v$ by $d(v)$, or by $d_H(v)$ when the graph $H$ needs to
be specified. For $X\subseteq V$, we write $G[X]=(X,E_G(X))$
for the subgraph induced by $X$.

Although we study simple undirected graphs, some proofs use
multigraphs and directed graphs. A \emph{multigraph} may have several distinct edges with the
same endpoints, called \emph{parallel edges}. A directed edge is
called an \emph{arc}. For an arc $u\to v$, $u$ is its \emph{tail}
and $v$ its \emph{head}: the arc leaves $u$ and enters $v$.
A \emph{directed multigraph} may have several distinct arcs with
the same tail and head, called \emph{parallel arcs}. We count all
parallel copies separately.

To keep this section self-contained, we recall the definitions of
arboricity and pseudoarboricity from \Cref{sec:intro}.

A \emph{forest} is an acyclic graph. The \emph{arboricity}
$a(G)$ is the minimum number of forests into which $E$ can be
partitioned. We represent each forest by its edge set and write
\begin{equation*}
  \mathcal F=(F_1,\ldots,F_k),
  \qquad E=F_1\dotcup\cdots\dotcup F_k
\end{equation*}
for a partition into $k$ forests. Here $\dotcup$ denotes disjoint
union, and each $(V,F_i)$ is a forest. Empty forests and isolated
vertices are implicit. For $m\geq 1$, the Nash--Williams
characterization~\cite{NashWilliams} gives
\begin{equation*}
  a(G)=\max_{\substack{X\subseteq V\\|X|\geq 2}}
    \left\lceil\frac{|E_G(X)|}{|X|-1}\right\rceil.
\end{equation*}
For an edgeless graph, $a(G)=0$.

The \emph{density} of a nonempty vertex set $X$ is
$\frac{|E_G(X)|}{|X|}$. We use the integer density parameter
\begin{equation*}
  \rho(G)=\left\lceil\max_{\varnothing\neq X\subseteq V}
    \frac{|E_G(X)|}{|X|}\right\rceil,
\end{equation*}
with $\rho(G)=0$ when $V=\varnothing$. A \emph{pseudoforest}
is a graph in which each connected component contains at most
one cycle. The \emph{pseudoarboricity} is the minimum number of
pseudoforests needed to partition $E$. The standard density
characterization states that this minimum equals $\rho(G)$.
Thus we also call $\rho(G)$ the pseudoarboricity of $G$.
We write $\mathcal P=(P_1,\ldots,P_k)$ for a pseudoforest
partition, with $E=P_1\dotcup\cdots\dotcup P_k$.

\begin{lemma}[Density and arboricity]
\label{lem:prelim-density-arboricity}
Every simple graph $G$ satisfies
\begin{equation*}
  \rho(G)\leq a(G)\leq\rho(G)+1.
\end{equation*}
\end{lemma}
\begin{proof}
Every forest is a pseudoforest, giving the first inequality.
For the second, let $k=\rho(G)$. For any $X\subseteq V$ with $|X|\geq 2$,
if $|X|\geq k+1$, the density bound gives
\begin{equation*}
  |E_G(X)|\leq k|X|\leq(k+1)(|X|-1).
\end{equation*}
Otherwise $|X|\leq k$, and simplicity gives
$|E_G(X)|\leq\binom{|X|}{2}\leq k(|X|-1)$.
The Nash--Williams characterization therefore implies
$a(G)\leq k+1$.
\end{proof}

An \emph{orientation} $\mathcal O$ of $G$ assigns a direction to every
edge of $G$. We view $\mathcal O$ as a directed graph on $V$ whose arcs
are the edges of $G$, each directed as assigned. For a vertex $v$, its outdegree $d^+_{\mathcal O}(v)$ counts the arcs with tail $v$,
and its indegree $d^-_{\mathcal O}(v)$ counts the arcs with head $v$. Its out-neighborhood
$N^+_{\mathcal O}(v)$ consists of the vertices $u$ with
$v\to u$ in $\mathcal O$, while its in-neighborhood
$N^-_{\mathcal O}(v)$ consists of those with $u\to v$.
An orientation $\mathcal O$ is \emph{$k$-out} if
$d^+_{\mathcal O}(v)\leq k$ for every vertex $v$.
We omit the subscript when the orientation is clear.

\begin{lemma}[Orientations and pseudoforests]
\label{lem:prelim-orientation-pseudoforest}
For every nonnegative integer $k$, a graph $G$ admits a $k$-out
orientation if and only if its edges can be partitioned into
$k$ pseudoforests. In particular,
\begin{equation*}
  \rho(G)=\min_{\mathcal O}\max_{v\in V}d^+_{\mathcal O}(v).
\end{equation*}
\end{lemma}
\begin{proof}
Given a $k$-out orientation, assign distinct colors
in $[k]$ to the outgoing edges of each vertex. In each color,
every vertex has outdegree at most one, so every connected
component has at most as many edges as vertices and is a
pseudoforest. Conversely, orient each tree component of a
pseudoforest toward a chosen root. In a component with a cycle,
orient the cycle cyclically and all remaining edges toward it.
Each vertex then has outdegree at most one. Applying this to
each of $k$ pseudoforests gives a $k$-out orientation. Taking
the minimum feasible $k$ proves the identity.
\end{proof}

\begin{proposition}[Computing a minimum orientation]
\label{prop:prelim-minimum-orientation}
For a graph with $m$ edges, the value $k=\rho(G)$ and a
$k$-out orientation can be computed deterministically in
$m^{1+o(1)}$ time. Testing whether a given outdegree bound
is feasible requires one exact maximum-flow computation on a
network with $O(m)$ vertices and arcs.
\end{proposition}
\begin{proof}
We binary search for the minimum integer $k\in\{0,\ldots,m\}$
such that $G$ admits a $k$-out orientation. Feasibility is
monotone: a $k$-out orientation is also a $(k+1)$-out
orientation, and $k=m$ is always feasible. To test a candidate
$k$, we construct the following flow network.

Create a source $s$, a sink $t$, one node $z_e$ for every edge
$e=uv$, and one node for each nonisolated vertex of $G$. Add
the arcs
\begin{equation*}
\begin{array}{c|c}
  \text{arc} & \text{capacity}\\ \hline
  s\to z_e & 1\\
  z_e\to u,\ z_e\to v & 1\\
  v\to t & k.
\end{array}
\end{equation*}
\Cref{fig:orientation-flow} illustrates the construction.

The idea is to send one unit of flow for each edge and use
that unit to choose the edge's tail. In an integral flow of
value $m$, every arc $s\to z_e$ carries one unit, since these
$m$ arcs have total capacity $m$. At $z_e$, this unit must
continue to exactly one endpoint of $e=uv$: routing it along
$z_e\to u$ means orienting $e$ as $u\to v$, while routing it
along $z_e\to v$ means orienting $e$ as $v\to u$. Thus the
flow through $v\to t$ counts exactly the edges oriented out
of $v$, and its capacity $k$ enforces the outdegree bound.

Conversely, any $k$-out orientation gives a flow of value $m$:
for every edge $u\to v$, send one unit along
$s\to z_{uv}\to u\to t$. The outdegree bound ensures that
the arcs into $t$ respect their capacities. Since all
capacities are integral, an integral maximum flow exists.
Hence a $k$-out orientation exists if and only if the
maximum-flow value is $m$.

Binary search therefore finds the minimum feasible $k$ using
$O(\log m)$ maximum-flow calls. Retain an integral flow of
value $m$ at this bound and read off the orientation as above.
\Cref{lem:prelim-orientation-pseudoforest} gives $k=\rho(G)$.

Each network has $O(m)$ vertices and arcs, with integral capacities at most
$m$. Exact maximum flow on such a network takes
$m^{1+o(1)}$ deterministic time~\cite{van2023deterministic}.
The logarithmic number of calls is absorbed into this bound.
Extracting the orientation takes $O(m)$ time.
\end{proof}

Thus, by \Cref{lem:prelim-orientation-pseudoforest}, we can compute
both the pseudoarboricity $k=\rho(G)$ and a corresponding
$k$-out orientation in $m^{1+o(1)}$ time.

\begin{figure}[htbp]
\centering
\begingroup
\tikzset{
  orientation node/.style={circle,draw=black,fill=white,
    inner sep=0pt,minimum size=21pt},
  orientation edge node/.style={draw=black,rounded corners=2pt,
    fill=blue!5,inner sep=5pt},
  orientation arc/.style={paper graph arc,draw=black!45},
  orientation flow/.style={paper forest arc,draw=blue!75!black},
  orientation capacity/.style={font=\scriptsize,fill=white,
    inner sep=1.5pt,text=black},
  orientation heading/.style={font=\scriptsize,text=black!65}
}
\begin{subfigure}[b]{0.24\linewidth}
  \centering
  \begin{tikzpicture}[x=1.25cm,y=1cm,font=\small]
    \path[use as bounding box] (-1.4,-1.55) rectangle (1.4,2);
    \node[orientation node] (u) at (-1,0) {$u$};
    \node[orientation node] (v) at (0,0) {$v$};
    \node[orientation node] (w) at (1,0) {$w$};
    \draw[orientation flow] (u) -- (v);
    \draw[orientation flow] (v) -- (w);
  \end{tikzpicture}
  \caption{A $1$-out orientation.}
\end{subfigure}\hfill
\begin{subfigure}[b]{0.72\linewidth}
  \centering
  \begin{tikzpicture}[x=1.4cm,y=1cm,font=\small]
    \path[use as bounding box] (-0.45,-1.55) rectangle (6.65,2);
    \node[orientation heading] at (0,1.75) {source};
    \node[orientation heading] at (2,1.75) {edge nodes};
    \node[orientation heading] at (4.2,1.79) {vertices};
    \node[orientation heading] at (6.2,1.75) {sink};
    \node[orientation node] (s) at (0,0) {$s$};
    \node[orientation edge node] (ab) at (2,0.8) {$z_{uv}$};
    \node[orientation edge node] (bc) at (2,-0.8) {$z_{vw}$};
    \node[orientation node] (a) at (4.2,1.25) {$u$};
    \node[orientation node] (b) at (4.2,0) {$v$};
    \node[orientation node] (c) at (4.2,-1.25) {$w$};
    \node[orientation node] (t) at (6.2,0) {$t$};

    \draw[orientation flow] (s) -- node[orientation capacity,above] {$1$} (ab);
    \draw[orientation flow] (s) -- node[orientation capacity,below] {$1$} (bc);
    \draw[orientation flow] (ab) -- node[orientation capacity,above] {$1$} (a);
    \draw[orientation arc] (ab) -- node[orientation capacity,above] {$1$} (b);
    \draw[orientation flow] (bc) -- node[orientation capacity,below] {$1$} (b);
    \draw[orientation arc] (bc) -- node[orientation capacity,below] {$1$} (c);
    \draw[orientation flow] (a) -- node[orientation capacity,above] {$k$} (t);
    \draw[orientation flow] (b) -- node[orientation capacity,above] {$k$} (t);
    \draw[orientation arc] (c) -- node[orientation capacity,below] {$k$} (t);
  \end{tikzpicture}
  \caption{The flow network for testing a candidate $k$.}
\end{subfigure}
\endgroup
\caption{The flow reduction for a minimum orientation. In (b),
arc labels are capacities, and the blue arcs carry one unit
each in a flow of value two for $k=1$; gray arcs carry no flow.
Routing $z_{uv}\to u$ selects $u$ as the tail of $uv$, and
routing $z_{vw}\to v$ selects $v$ as the tail of $vw$. The
corresponding orientation $u\to v\to w$ is shown in (a).}
\label{fig:orientation-flow}
\end{figure}

\FloatBarrier


\section{A constructive additive-one algorithm}
\label{sec:additive}

In this section, we prove \Cref{thm:additive} restated below.

\thmAdditive*

We provide \Cref{alg:additiveone} for proving the above theorem. The algorithm works in two stages. First, we convert a $k$-out orientation into $k$ pseudoforests. Second, we convert the resulting pseudoforests into $k+1$ forests by applying the algorithm of \citet{BlumenstockFischer}. Their analysis shows that the algorithm runs in $O(mk\log k+m\log n)$ time. We analyze their algorithm more carefully to show that it runs in $O(m\log n+nk+Zk\log n)$ time, where $Z$ is the total number of cycles in the pseudoforests. The overall additive-one algorithm of \citet{BlumenstockFischer} uses an arbitrary conversion of $k$-out orientations to pseudoforests and hence can only use the trivial $Z=O(m)$ bound, which leads to their $O(mk\log k+m\log n)$ running time. We, however, show a way to find pseudoforests with few cycles, specifically with $O(n\log k)$ total cycles. This gives an $O(m\log n+nk\log k\log n)$ running time for the conversion.

\begin{algorithm}[H]
\caption{$\AdditiveOne(G,k,\mathcal O)$}
\label{alg:additiveone}
\KwIn{Simple graph $G$, $k\geq 1$, and a $k$-out orientation $\mathcal O$}
\KwOut{A partition of $E(G)$ into $k+1$ forests}
$\mathcal O'\gets\mathcal O$\; \nllabel{line:additive-pseudoforests-init}
\For{$d=k,k-1,\ldots,1$}{
  $S\gets\textsc{ExtractOnePseudoforest}(\mathcal O',d)$ from \Cref{alg:extractonepseudoforest}\nllabel{line:additive-extract}\;
  $P_{k-d+1}\gets S$\nllabel{line:additive-pseudoforests}\;
  Delete all edges of $S$ from $\mathcal O'$\;
}\nllabel{line:additive-pseudoforests-end}
Use \Cref{prop:conversion} to convert $P_1,\ldots,P_k$ into $F_1,\ldots,F_{k+1}$\nllabel{line:additive-conversion}\;
\Return $(F_1,\ldots,F_{k+1})$\;
\end{algorithm}

\subsection{Constructing pseudoforests with few cycles}
\label{sec:lowcycles}

To find pseudoforests with a small number of cycles, we build one pseudoforest at a time by selecting at most one outgoing edge for each vertex. For the first pseudoforest, we process the vertices of outdegree $k$ one by one. At each vertex, we greedily choose an outgoing edge that does not create a cycle with the edges already chosen for this pseudoforest. If every choice creates a cycle, we select any outgoing edge. When this happens, all $k$ choices lead to vertices already connected to the selecting vertex by previously chosen edges. Since the graph is simple, these choices have $k$ distinct heads, so the component contains at least $k+1$ vertices. Each component of the resulting pseudoforest has at most one cycle, so there are at most $n/(k+1)$ cycles in the first pseudoforest.

Removing these edges leaves a $(k-1)$-out orientation: every vertex of outdegree $k$ loses one outgoing edge, and all other vertices already have outdegree at most $k-1$. Repeating the same construction gives at most $\sum_{d=1}^{k}n/(d+1)=O(n\log k)$ cycles in total. We explain below how to carry out these greedy choices in $O(m+nk\log n)$ total time.

We call each iteration of the loop in \Cref{alg:additiveone} a \emph{round}. In each round, let $\mathcal O'$ be the remaining $d$-out orientation of the edges not yet placed in a pseudoforest. Each vertex of outdegree exactly $d$ must select one outgoing edge, while vertices of smaller outdegree select none.

To make the greedy choices, we maintain the components of the selected edges. We call a component \emph{unfinished} if some vertex in it still needs to select an edge, and \emph{finished} otherwise. Each unfinished component is a tree directed toward its only vertex still needing to choose, which we call its root. Initially, vertices of outdegree $d$ are unfinished singletons, and all other vertices are finished singletons. We prove that the algorithm preserves this tree description in \Cref{lem:onepseudoforest}.

For an unfinished component with root $r$, scan the outgoing edges of $r$ in $\mathcal O'$ for an edge leading to another component and stop on the first such edge. If one exists, select it and merge the components. If the other component is finished, the union is finished too. If no edge leaves the component, select any outgoing edge of $r$. This closes a cycle and finishes the component. Thus every choice follows the greedy rule we mentioned above.

To obtain the whole pseudoforest, we work in phases. In each phase, we iterate over all unfinished components and try to choose an outgoing edge for its root. We process an unfinished component only if it has not merged during that phase. After a merge, mark the union with the current phase and postpone processing its root until the next phase. This way, the number of unfinished components halves at each phase. Furthermore, because the processed components of a phase are disjoint, the total running time in all phases is $O(n \log n)$.

\begin{algorithm}[H]
\caption{\textsc{ExtractOnePseudoforest}$(\mathcal O',d)$}
\label{alg:extractonepseudoforest}
\KwIn{A $d$-out orientation $\mathcal O'$ of a simple graph, with $d\geq 1$}
\KwOut{One outgoing edge from each vertex of outdegree $d$, and none from vertices of smaller outdegree}
$S\gets\varnothing$\;
Start with every vertex as a singleton, unfinished if its outdegree is $d$ and finished otherwise\;
\While{an unfinished component remains}{
  Start a new phase, marking all components as not merged in this phase\;
  \ForEach{vertex $r$}{
    $C\gets$ the current component containing $r$\;
    \If{$C$ is unfinished, has not merged in this phase, and $r$ is its root}{
      Scan the outgoing edges of $r$ for an edge leaving $C$ and stop on the first such edge\;
      \eIf{such an edge $r\to u$ exists}{
        $T\gets$ the component containing $u$\;
        Add $r\to u$ to $S$ and merge $C$ with $T$\;
        If $T$ was finished, mark the union finished\;
        Record that the union has merged in this phase\;
      }{
        Add any outgoing edge of $r$ to $S$\;
        mark $C$ finished\;
      }
    }
  }
}
\Return $S$\;
\end{algorithm}

We next prove formally that \Cref{alg:extractonepseudoforest} works correctly.

\begin{lemma}[One pseudoforest]
\label{lem:onepseudoforest}
\Cref{alg:extractonepseudoforest} selects one outgoing edge from each vertex of outdegree $d$ and none from the other vertices, in $O(n\log n)$ time. The selected edges form a pseudoforest with at most $n/(d+1)$ cycles.
\end{lemma}
\begin{proof}
Each unfinished component stays a tree directed toward its root, since an edge from one root into another unfinished tree gives a tree directed toward the other root. A root stops being a root once it selects an edge, and vertices of smaller outdegree start in finished components. So each vertex selects at most one edge, and only vertices of outdegree $d$ select. Thus the selected edges form a pseudoforest.

Next we bound the number of phases. If an unfinished component is skipped in a phase, it means that the component was merged with another component during this phase. This means that the number of unfinished components is reduced by at least half in each phase, and so there are $O(\log n)$ phases.

The root selects an edge inside its component only when none of its $d$ outgoing edges leaves the component. Since the graph is simple, their heads are $d$ distinct vertices other than the root. So the component has at least $d+1$ vertices when the cycle closes, and later merges only add vertices. The final cyclic components are disjoint and each contains one cycle, so there are at most $n/(d+1)$ cycles.

For the running time, a search in a component $C$ examines at most $|C|$ edges: at most $|C|-1$ can stay inside $C$, and the search stops at the first edge leaving it. After the search, $C$ finishes or merges. A finished component is never processed, and a merged one is not processed again in the same phase. So no component containing the vertices of $C$ searches again in that phase. Thus the searched components are disjoint, and a phase examines at most $n$ edges, besides scanning all vertices.

We use the list implementation of union-find with union by size. Each vertex stores its component representative, so component tests take constant time. A merge updates the representatives in the smaller list. Each update at least doubles the size of the vertex's component, so each vertex is updated $O(\log n)$ times during one call of \Cref{alg:extractonepseudoforest}. Each component also stores its root, whether it is finished, and its last merge phase, which take constant time to update per merge. Thus all merges cost $O(n\log n)$ time, and the $O(\log n)$ phases cost another $O(n\log n)$ total time for scans and edge tests.
\end{proof}

Using the correctness of \Cref{alg:extractonepseudoforest}, we can now show that the pseudoforests have a small numbers of cycles in total.

\begin{lemma}[Total number of cycles]
\label{lem:cycles}
Given a $k$-out orientation, the first stage of \Cref{alg:additiveone} in lines~\ref{line:additive-pseudoforests-init}--\ref{line:additive-pseudoforests-end} constructs $k$ pseudoforests with at most $O(n\log k)$ cycles in total, in $O(m+nk\log n)$ time.
\end{lemma}
\begin{proof}
Each round starts with a $d$-out orientation. By \Cref{lem:onepseudoforest}, every vertex of outdegree $d$ selects one outgoing edge. Removing these edges leaves a $(d-1)$-out orientation, since all other vertices already have outdegree at most $d-1$.

After round $1$, every vertex has outdegree zero, so the selected pseudoforests partition all edges. Summing the cycle bound from \Cref{lem:onepseudoforest} gives at most $n\sum_{d=1}^{k}1/(d+1)=O(n\log k)$ cycles.

The $k$ rounds take $O(nk\log n)$ total time by \Cref{lem:onepseudoforest}. We represent the graph in a list of outgoing edges per vertex. Delete a selected edge by replacing its list entry with the last entry in the list. All deletions together and building the outgoing lists take $O(n+m)$ time.
\end{proof}

\subsection{Converting pseudoforests into forests}
\label{sec:conversion}

We now focus on converting the pseudoforests $P_1,\ldots,P_k$ into forests in line~\ref{line:additive-conversion} of \Cref{alg:additiveone}. We use the algorithm of \citet{BlumenstockFischer}, but we must prove that it runs in $O(m\log n+nk+Zk\log n)$ time, where $Z$ is the total number of cycles in the pseudoforests.

\begin{proposition}[Pseudoforest conversion]
\label{prop:conversion}
Let $G$ be a simple graph with $n$ vertices and $m$ edges. Given a partition of $E(G)$ into $k$ pseudoforests with $Z$ cycles in total, the algorithm of \citet{BlumenstockFischer} converts it into $k+1$ forests in time $O(m\log n+nk+Zk\log n)$.
\end{proposition}

For completeness, we describe the conversion algorithm below. We prove the running-time bound of \Cref{prop:conversion} in \Cref{app:conversion}.

The algorithm starts as follows. For each $i$, remove one arbitrary edge from each cycle of $P_i$, and let $M_i$ be the set of removed edges. Each $M_i$ is a matching because the cycles of $P_i$ are vertex-disjoint. Also let $Y = M_1 \dotcup \dots \dotcup M_k$. Because we picked one edge from each cycle, we have $|Y|=Z$ as defined above. Set $F_i = P_i \setminus M_i$. The algorithm proceeds to make changes to the sets $F_1,\ldots,F_k$ and $M_1,\ldots,M_k$ until $Y$ becomes a forest, and then returns the resulting $k+1$ forest partition.

We call an edge \emph{surplus} if it belongs to $Y$, or equivalently, to some $M_i$. We view $Y$ as a graph on $V(G)$ and call it the \emph{surplus graph}. An edge has color $i$ if it belongs to $P_i=F_i\dotcup M_i$. Throughout the algorithm, $F_i$ remains a forest and $M_i$ contains exactly one cycle edge from each cyclic component of $P_i$, and no other edges. We call such a partition \emph{valid}. In particular, the endpoints of each edge of $M_i$ lie in one connected component, or more specifically, in one tree of $F_i$, and each tree of $F_i$ contains the endpoints of at most one edge of $M_i$.

The algorithm uses the following two operations.
\begin{enumerate}
  \item \emph{Swap the color of two surplus edges that share an endpoint.} Let $e_1=uv\in M_i$ and $e_2=vw\in M_j$. Since each $M_i$ is a matching, $i\neq j$. Give $e_1$ color $j$ and $e_2$ color $i$. Remove $e_1$ from $M_i$ and $e_2$ from $M_j$. If $v$ and $w$ lie in the same tree of $F_i$, put $e_2$ into $M_i$. Otherwise, put it into $F_i$. Likewise, put $e_1$ into $M_j$ if $u$ and $v$ lie in the same tree of $F_j$, and into $F_j$ otherwise.

  To see why this preserves validity for color $i$, consider the tree of $F_i$ containing $u$ and $v$. Removing $e_1$ leaves this tree with no surplus edge of color $i$. If $w$ lies in this tree, $e_2$ replaces $e_1$ as its surplus edge. Otherwise, $e_2$ joins it to another tree, which has at most one surplus edge of color $i$. The joined tree therefore still has at most one such edge. The same argument applies to color $j$.

  \item \emph{Exchange a surplus edge with a forest edge.} Let $uv\in M_i$. By validity, adding $uv$ to $F_i$ creates a cycle. Let $vw$ be the other edge of this cycle incident to $v$. Move $uv$ from $M_i$ to $F_i$ and move $vw$ from $F_i$ to $M_i$. Both edges keep color $i$, and $P_i$ stays unchanged. Removing $vw$ breaks the cycle without changing the tree's vertex set, and $vw$ becomes the surplus edge of that component. Thus the partition remains valid.
\end{enumerate}

We next explain how the algorithm uses these operations. We call $Y$ \emph{colorful} if each connected component contains at most one edge of each color. A colorful component has at most $k$ edges, so scanning the whole component takes $O(k)$ time. To make $Y$ colorful, suppose a component has two edges of color $i$, and take a path whose first and last edges have color $i$. Apply the swap operation to consecutive edges along this path, moving color $i$ from the first edge toward the last. Stop when an edge moves into a forest. This happens before the two surplus edges of color $i$ become incident, since validity requires $M_i$ to remain a matching. Thus $Y$ loses at least one edge. Repeating this process makes $Y$ colorful.

To find a repeated color, search a component until a color appears twice. This takes at most $k+1$ edge visits, and the visited edges contain a path between the two edges of that color. The path therefore has length $O(k)$. Thus it takes $O(k)$ time to find the path and at most $O(k)$ swaps to remove a surplus edge.

Now suppose $Y$ is colorful but still contains a cycle. Pick an edge $uv$ on this cycle, and let $C$ be the entire component of $Y$ that contains it. \citet{BlumenstockFischer} show that we can choose a color $i$ present in $C$ such that no edge of $F_i$ joins $v$ to another vertex of $C$. We include the relevant lemma and a short proof in the appendix for completeness (\Cref{lem:cycle-exchange-choice}). Move color $i$ to $uv$ using the first operation along a sequence of incident edges. If an edge moves into a forest, $Y$ loses an edge and we restart. Otherwise, the second operation at $v$ replaces $uv$ with an edge $vw$ from $F_i$. By our choice of $i$, $w$ lies outside $C$, so the operation merges $C$ with another component of $Y$. Thus $vw$ is a bridge in $Y$, while the removed edge $uv$ was on a cycle. We call this a cycle exchange. The exchange may introduce repeated colors in the joined component, which we remove by the path-swapping process just described.

We call an \emph{attempt} one run of swaps that moves a color along a path, either to remove a repeated color or to move a chosen color to a cycle edge. Each attempt uses $O(k)$ swaps and ends by removing a surplus edge or by performing a cycle exchange. There are at most $Z$ surplus-edge removals, since no operation increases the number of edges in $Y$. There are also at most $Z$ cycle exchanges, since each decreases the number of edges lying on cycles of $Y$, and surplus-edge removals cannot increase this number. Including initialization and searches, the total running time is $O(m\log n+nk+Zk\log n)$. We prove the running-time bound of \Cref{prop:conversion} in \Cref{app:conversion}. When no cycle remains in $Y$, we return $F_1,\ldots,F_k,Y$.

\subsection[Proof of Theorem~\ref*{thm:additive}]{Proof of \Cref{thm:additive}}
\label{sec:additivewrapper}

We now combine the two bounds to prove \Cref{thm:additive}.

\begin{proof}[Proof of \Cref{thm:additive}]
By \Cref{lem:cycles}, the first stage of \Cref{alg:additiveone} constructs $k$ pseudoforests in $O(m+nk\log n)$ time, with $Z=O(n\log n)$ cycles in total. By \Cref{prop:conversion}, converting them into $k+1$ forests takes
\begin{equation*}
  O(m\log n+nk+Zk\log n)
  =O(m\log n+nk\log^2 n)
\end{equation*}
time. This also covers the construction time, proving \Cref{thm:additive}.
\end{proof}


\section{Star Insertion}
\label{sec:star-insertion}

In this section, we prove \Cref{thm:star}, restated below, which lets us delete an arbitrary vertex, partition the remaining edges into forests, and then reinsert the deleted vertex.
Given the partition of the smaller graph, we either construct a partition of the whole graph or correctly report that none exists.

\thmStarInsertion*

Let $H$ be a simple graph with a specified vertex $r$, and suppose that $E(H-r)$ is partitioned into $k\geq 1$ forests $F_1,\ldots,F_k$, where an edge in $F_i$ has color $i$.
Our task is to insert the edges incident to $r$: we may change the colors of the existing edges, but the edges of each color must finally form a forest.
Write $U=\{rv\in E(H)\}$.
We call the edges of $H-r$ original edges and those of $U$ star edges.

\subsection{Adding one edge}
\label{sec:star-design}

We first consider the case in which only one star edge remains to be inserted.
This simpler setting will explain color-shifts and the choice of costs.
We then use a flow network to coordinate several insertions at once.

Suppose that every star edge except $rv$ has already been inserted.
Let $F_1,\ldots,F_k$ be forests partitioning $E(H)\setminus\{rv\}$.
We keep the colors of the inserted star edges fixed.
We first give each edge in the forests a direction.
We want every vertex to have at most one outgoing edge of each color, with $r$ a sink.
To obtain this orientation, root the component containing $r$ at $r$ in each $F_i$, root every other component arbitrarily, and direct edges toward the roots.
Thus each vertex has at most one outgoing edge of each color.

\begin{definition}[Valid orientation and free colors]
\label{def:star-rooted}
\normalfont
Suppose some edges of $H$ have colors in $[k]$.
An orientation of the colored edges is a \emph{valid orientation} with respect to $r$ if every vertex has at most one outgoing edge of each color, and $r$ has no outgoing edges.
We call a vertex \emph{free of color $i$} if it has no outgoing edge of color $i$.
In particular, $r$ is free of every color.
\end{definition}

It is important to note that in a valid orientation, the edges of a given color are allowed to contain cycles.
Now consider the uncolored edge $rv$.
If $v$ is free of some color $i$, assigning $rv$ color $i$ and direction $v \to r$ is sufficient.
Both $v$ and $r$ are roots in $F_i$, and thus they lie in different components.
Otherwise, we can shift the color of an outgoing original edge $v \to u$ at $v$ onto $rv$, leaving $vu$ uncolored.
This way, we might be able to color $vu$ if $u$ is free of a color.
\Cref{def:star-color-shift}, given below, describes this change for a general uncolored edge.

\begin{definition}[Color-shift]
\label{def:star-color-shift}
\normalfont
Let the colored edges have a valid orientation with respect to $r$, as in \Cref{def:star-rooted}.
A \emph{color-shift} is defined as follows.
Suppose $xy$ is an uncolored edge and the outgoing edge of color $c$ at $y$ is an original edge $y \to z$ (so $z \neq r$).
We remove the color and direction from $yz$, then color $xy$ with $c$ and direct it as $y\to x$.
\end{definition}

After a color-shift, $y$ still has one outgoing edge of color $c$, while $yz$ replaces $xy$ as the uncolored edge (see \Cref{fig:color-shift-single}).
The following lemma shows that a color-shift preserves a valid orientation and the free colors at every vertex.

\begin{lemma}[Preserving the valid orientation]
\label{obs:star-shift-orientation}
Suppose the colored edges of $H$ have a valid orientation with respect to $r$.
Applying a color-shift of color $c$ to an uncolored edge $xy$ at $y$ preserves a valid orientation with respect to $r$ and does not change which colors are free at any vertex.
\end{lemma}
\begin{proof}
A color-shift removes one outgoing edge of color $c$ at $y$ and replaces it with another.
All other outdegrees are unchanged.
Thus the two conditions in \Cref{def:star-rooted} still hold, and the free colors at every vertex remain the same.
\end{proof}

Starting from $rv$, we can repeatedly apply color-shifts along outgoing original edges.
Our goal is to reach a vertex other than $r$ with a free color.
We now define a sequence that reaches such a vertex.

\begin{definition}[Complete sequence of color-shifts]
\label{def:star-shift-sequence}
\normalfont
Let the colored edges have a valid orientation with respect to $r$, and let $rv$ be an uncolored edge.
A \emph{complete sequence of color-shifts} for $rv$ is a list of vertices with chosen colors,
\begin{equation*}
  \mathcal S=((v_1,c_1),\ldots,(v_\ell,c_\ell)),
\end{equation*}
where $\ell\geq 1$, $v_1=v$, and $c_j\in[k]$ for every $j$, with the following properties:
\begin{enumerate}
\item For each $j<\ell$, after the preceding shifts, the outgoing edge of color $c_j$ at $v_j$ is an original edge $v_j\to v_{j+1}$ (hence $v_{j+1}\neq r$).
We apply a color-shift at $v_j$, giving the current uncolored edge color $c_j$ directed out of $v_j$ and leaving $v_jv_{j+1}$ uncolored.
The sequence then continues at $v_{j+1}$.
\item After these shifts, the last vertex $v_\ell$ differs from $r$ and is free of color $c_\ell$.
\item We give the current uncolored edge color $c_\ell$ and direct it out of $v_\ell$.
\end{enumerate}
\end{definition}

The following lemma shows that a complete sequence colors $rv$ while preserving a valid orientation.

\begin{lemma}[Preservation along a complete sequence]
\label{obs:star-shift-sequence}
Suppose the colored edges of $H$ have a valid orientation with respect to $r$ and $rv$ is an uncolored edge.
Applying a complete sequence of color-shifts $\mathcal S=((v_1,c_1),\ldots,(v_\ell,c_\ell))$ for $rv$ colors $rv$, keeps all previously colored edges colored, and preserves a valid orientation with respect to $r$.
\end{lemma}
\begin{proof}
By \Cref{obs:star-shift-orientation}, every shift preserves all outdegrees.
Each shift replaces one uncolored edge by another.
The final step colors that edge and increases one outdegree from zero to one at $v_\ell\neq r$.
Thus each vertex still has at most one outgoing edge of each color, and $r$ still has no outgoing edges.
\end{proof}

Recalling our insertion method, we start with $rv$ uncolored: when $v$ has no free color, we choose a color $j$ and let $v \to u$ be the outgoing original edge of color $j$ at $v$.
Applying a color-shift assigns $rv$ color $j$ and direction $v \to r$, leaving $vu$ uncolored; if the outgoing edge of color $i$ at $u$ is an original edge $u \to z$, a second color-shift assigns $vu$ color $i$ and direction $u \to v$, leaving $uz$ uncolored.
We continue applying color-shifts in this way.
If the uncolored edge $ab$ reaches a vertex $b \neq r$ free of some color, we assign $ab$ that color and direct it as $b \to a$ (see \Cref{fig:color-shift-path}).
Throughout this process, every vertex has outdegree at most one in each color.

\usetikzlibrary{backgrounds}
\begingroup
\definecolor{ShiftRed}{HTML}{B64040}
\definecolor{ShiftBlue}{HTML}{265FA8}
\definecolor{ShiftGreen}{rgb}{0.1333,0.5451,0.1333}
\setlength{\intextsep}{7pt}
\begin{figure}[H]
\centering
\tikzset{
  csvertex/.style={paper graph vertex},
  csroot/.style={csvertex,double,double distance=0.5pt,minimum size=4.4mm,font=\fontsize{8.5}{10}\selectfont\bfseries},
  csdots/.style={circle,fill=white,minimum size=4.1mm,inner sep=1pt,font=\fontsize{8.5}{10}\selectfont},
  csgray/.style={paper pending edge,line width=0.65pt},
  csred/.style={paper forest arc,line width=0.85pt,-{Latex[length=1.65mm,width=1.1mm]},draw=ShiftRed},
  csblue/.style={paper forest arc,line width=0.85pt,-{Latex[length=1.65mm,width=1.1mm]},draw=ShiftBlue},
  csgreen/.style={paper forest arc,line width=0.85pt,-{Latex[length=1.65mm,width=1.1mm]},draw=ShiftGreen},
  csnewred/.style={paper highlighted arc,line width=1.9pt,-{Latex[length=2.2mm,width=1.55mm]},draw=ShiftRed},
  csnewblue/.style={paper highlighted arc,line width=1.9pt,-{Latex[length=2.2mm,width=1.55mm]},draw=ShiftBlue},
  csnewgreen/.style={paper highlighted arc,line width=1.9pt,-{Latex[length=2.2mm,width=1.55mm]},draw=ShiftGreen},
  csstep/.style={paper step,line width=0.65pt,draw=black!45}
}
\begin{subfigure}[b]{0.32\textwidth}
\centering
\begin{tikzpicture}[x=1.15cm,y=0.95cm]
  \path (0,0.25) (0,-2.25);
  \node[csvertex] (x1) at (0,-0.5) {$x$};
  \node[csvertex] (y1) at (1,-0.5) {$y$};
  \node[csvertex] (z1) at (2,-0.5) {$z$};
  \node[csvertex] (x2) at (0,-1.5) {$x$};
  \node[csvertex] (y2) at (1,-1.5) {$y$};
  \node[csvertex] (z2) at (2,-1.5) {$z$};

  \begin{scope}[on background layer]
    \draw[csgray]   (x1) -- (y1);
    \draw[csred]    (y1) -- (z1);
    \draw[csstep]   (1,-0.78) -- (1,-1.22);
    \draw[csnewred] (y2) -- (x2);
    \draw[csgray]   (y2) -- (z2);
  \end{scope}
\end{tikzpicture}
\subcaption{Color-shift.}
\label{fig:color-shift-single}
\end{subfigure}%
\hspace{0.05\textwidth}%
\begin{subfigure}[b]{0.58\textwidth}
\centering
\begin{tikzpicture}[x=1.15cm,y=0.95cm]
  \path (0,0.25) (0,-2.25);
  \node[csroot]   (r1) at (0,0) {$r$};
  \node[csvertex] (v1) at (1,0) {$v$};
  \node[csvertex] (u1) at (2,0) {$u$};
  \node[csdots]   (d1) at (3,0) {$\cdots$};
  \node[csvertex] (w1) at (4,0) {$a$};
  \node[csvertex] (s1) at (5,0) {$b$};

  \node[csroot]   (r2) at (0,-1) {$r$};
  \node[csvertex] (v2) at (1,-1) {$v$};
  \node[csvertex] (u2) at (2,-1) {$u$};
  \node[csdots]   (d2) at (3,-1) {$\cdots$};
  \node[csvertex] (w2) at (4,-1) {$a$};
  \node[csvertex] (s2) at (5,-1) {$b$};

  \node[csroot]   (r3) at (0,-2) {$r$};
  \node[csvertex] (v3) at (1,-2) {$v$};
  \node[csvertex] (u3) at (2,-2) {$u$};
  \node[csdots]   (d3) at (3,-2) {$\cdots$};
  \node[csvertex] (w3) at (4,-2) {$a$};
  \node[csvertex] (s3) at (5,-2) {$b$};

  \begin{scope}[on background layer]
    \draw[csgray]     (r1) -- (v1);
    \draw[csred]      (v1) -- (u1);
    \draw[csblue]     (u1) -- (d1);
    \draw[csgreen]    (d1) -- (w1);
    \draw[csred]      (w1) -- (s1);

    \draw[csstep]     (2.5,-0.28) -- (2.5,-0.72);

    \draw[csnewred]   (v2) -- (r2);
    \draw[csnewblue]  (u2) -- (v2);
    \draw[csnewgreen] (d2) -- (u2);
    \draw[csnewred]   (w2) -- (d2);
    \draw[csgray]     (w2) -- (s2);

    \draw[csstep]     (2.5,-1.28) -- (2.5,-1.72);

    \draw[csred]      (v3) -- (r3);
    \draw[csblue]     (u3) -- (v3);
    \draw[csgreen]    (d3) -- (u3);
    \draw[csred]      (w3) -- (d3);
    \draw[csnewgreen] (s3) -- (w3);
  \end{scope}
\end{tikzpicture}
\subcaption{Complete sequence of color-shifts.}
\label{fig:color-shift-path}
\end{subfigure}
\caption{Color-shift and complete sequence of color-shifts (gray dashed edges are uncolored).}
\label{fig:color-shift}
\end{figure}
\endgroup

We first show that if $H$ can be partitioned into $k$ forests, a complete sequence of color-shifts exists.
Otherwise, the vertices reachable from $v$ along the input directions in $H-r$, together with $r$, form a set $X$ spanning more than $k(|X|-1)$ edges, violating the Nash--Williams condition.

\begin{lemma}[Existence of a complete sequence]
\label{lem:star-sequence-exists}
Suppose $E(H)\setminus\{rv\}$ is partitioned into $k$ forests $F_1,\ldots,F_k$ oriented with a valid orientation with respect to $r$.
If $H$ can be partitioned into $k$ forests, then there exists a complete sequence of color-shifts $\mathcal S=((v_1,c_1),\ldots,(v_\ell,c_\ell))$ for $rv$.
The sequence can be chosen so that $v=v_1\to v_2\to\cdots\to v_\ell$ is a simple directed path in the input orientation of $H-r$.
\end{lemma}

\begin{proof}[Proof]
Suppose otherwise, and let $X$ consist of $r$ and all vertices reachable from $v$ along the input directed edges in $H-r$.
We first show that no vertex in $X\setminus\{r\}$ is free of any color.
If a reached vertex $b$ were free of color $c$, take a simple directed path from $v$ to $b$ and shift along its edges in order.
At each vertex, the next edge still has its input color and direction, because the path has not visited that vertex before.
The shifts leave $b$ free of color $c$ by \Cref{obs:star-shift-orientation}, giving a complete sequence, contrary to our assumption.

In the input orientation, each $y\in X\setminus\{r\}$ therefore has one outgoing edge $y\to z$ of every color.
Each such edge ends in $X$: either $z=r$, or $z$ is reachable from $v$ via $y$.
Thus all $k(|X|-1)$ outgoing edges from $X \setminus \{r\}$ lie inside $X$.
These edges are distinct: each is counted at its tail and in its own color.
Including $rv$ gives $|E_H(X)| \geq k(|X|-1) + 1$, contradicting the Nash--Williams condition:
\begin{equation*}
  a(H) \geq \left\lceil\frac{|E_H(X)|}{|X|-1}\right\rceil \geq \left\lceil\frac{k(|X|-1)+1}{|X|-1}\right\rceil = k+1 > k.
\end{equation*}
\end{proof}

However, applying a complete sequence of color-shifts can still create directed cycles, even though it guarantees outdegree at most one in each color (see \Cref{fig:cycle-comp-bad}).
We next provide an efficient way to find a complete sequence of color-shifts so that the resulting coloring contains no directed cycle.

\usetikzlibrary{backgrounds}
\begingroup
\definecolor{ShiftRed}{HTML}{B64040}
\definecolor{ShiftBlue}{HTML}{265FA8}
\setlength{\intextsep}{7pt}
\begin{figure}[H]
\centering
\tikzset{
  csvertex/.style={paper graph vertex},
  csroot/.style={csvertex,double,double distance=0.5pt,minimum size=4.4mm,font=\fontsize{8.5}{10}\selectfont\bfseries},
  csgray/.style={paper pending edge,line width=0.65pt},
  csred/.style={paper forest arc,line width=0.85pt,-{Latex[length=1.65mm,width=1.1mm]},draw=ShiftRed},
  csblue/.style={paper forest arc,line width=0.85pt,-{Latex[length=1.65mm,width=1.1mm]},draw=ShiftBlue},
  csnewred/.style={paper highlighted arc,line width=1.9pt,-{Latex[length=2.2mm,width=1.55mm]},draw=ShiftRed},
  csnewblue/.style={paper highlighted arc,line width=1.9pt,-{Latex[length=2.2mm,width=1.55mm]},draw=ShiftBlue}
}
\begin{subfigure}[b]{0.31\textwidth}
\centering
\begin{tikzpicture}[x=1.25cm,y=1.05cm]
  \node[csroot]   (r) at (1,0)   {$r$};
  \node[csvertex] (u) at (0,0.6) {$u$};
  \node[csvertex] (v) at (0,1.5) {$v$};
  \node[csvertex] (w) at (1,1.5) {$a$};
  \node[csvertex] (s) at (2,1.5) {$b$};
  \node[csvertex] (z) at (1.5,2.3) {$z$};

  \begin{scope}[on background layer]
    \draw[csgray] (v) -- (r);
    \draw[csblue] (v) -- (u);
    \draw[csblue] (u) -- (r);
    \draw[csred]  (v) -- (w);
    \draw[csred]  (w) -- (s);
    \draw[csred]  (s) -- (r);
    \draw[csblue] (w) -- (z);
    \draw[csblue] (z) -- (s);
  \end{scope}
\end{tikzpicture}
\subcaption{Input ($b$ is blue-free).}
\label{fig:cycle-comp-input}
\end{subfigure}%
\hfill%
\begin{subfigure}[b]{0.31\textwidth}
\centering
\begin{tikzpicture}[x=1.25cm,y=1.05cm]
  \node[csroot]   (r) at (1,0)   {$r$};
  \node[csvertex] (u) at (0,0.6) {$u$};
  \node[csvertex] (v) at (0,1.5) {$v$};
  \node[csvertex] (w) at (1,1.5) {$a$};
  \node[csvertex] (s) at (2,1.5) {$b$};
  \node[csvertex] (z) at (1.5,2.3) {$z$};

  \begin{scope}[on background layer]
    \draw[csnewred]  (v) -- (r);
    \draw[csblue]    (v) -- (u);
    \draw[csblue]    (u) -- (r);
    \draw[csnewred]  (w) -- (v);
    \draw[csnewblue] (s) -- (w);
    \draw[csred]     (s) -- (r);
    \draw[csblue]    (w) -- (z);
    \draw[csblue]    (z) -- (s);
  \end{scope}
\end{tikzpicture}
\subcaption{Bad shifts (blue cycle).}
\label{fig:cycle-comp-bad}
\end{subfigure}%
\hfill%
\begin{subfigure}[b]{0.31\textwidth}
\centering
\begin{tikzpicture}[x=1.25cm,y=1.05cm]
  \node[csroot]   (r) at (1,0)   {$r$};
  \node[csvertex] (u) at (0,0.6) {$u$};
  \node[csvertex] (v) at (0,1.5) {$v$};
  \node[csvertex] (w) at (1,1.5) {$a$};
  \node[csvertex] (s) at (2,1.5) {$b$};
  \node[csvertex] (z) at (1.5,2.3) {$z$};

  \begin{scope}[on background layer]
    \draw[csnewred]  (v) -- (r);
    \draw[csblue]    (v) -- (u);
    \draw[csblue]    (u) -- (r);
    \draw[csnewblue] (w) -- (v);
    \draw[csred]     (w) -- (s);
    \draw[csred]     (s) -- (r);
    \draw[csnewblue] (z) -- (w);
    \draw[csnewblue] (s) -- (z);
  \end{scope}
\end{tikzpicture}
\subcaption{Good shifts (forests).}
\label{fig:cycle-comp-good}
\end{subfigure}
\caption{Inserting $rv$ via color-shifts (thick edges are shifted or newly colored).}
\label{fig:cycle-comp}
\end{figure}
\endgroup

We build a directed graph $Q$ with $O(nk)$ vertices and $O(m+nk)$ \flowedge{}s, all with nonnegative costs, as defined in \Cref{def:star-auxiliary-graph} below.

\begin{definition}[Auxiliary graph]
\label{def:star-auxiliary-graph}
\normalfont
We introduce an auxiliary directed graph $Q$ with a vertex $q_x$ for each $x \in V(H) \setminus \{r\}$ and a vertex $q_{x, i}$ for each $x \in V(H) \setminus \{r\}$ and $i \in [k]$.
Let $d_i(x)$ be the depth of $x$ in $F_i$, with roots at depth zero, and define
\begin{equation*}
  w_i(x) = 2m - d_i(x) > 0.
\end{equation*}
Thus, for any edge $x \to y$ of color $i$, $d_i(x) = d_i(y) + 1$ and $w_i(x) = w_i(y) - 1$.
The \flowedge{}s of $Q$ and their costs are defined as follows:
\begin{itemize}
\item For each original edge $x \to y$ of color $i$ in $F_i$, add an \flowedge{} $q_{x, i} \to q_{y, i}$ of cost $0$.
\item For each vertex $x \neq r$ and each color $i \in [k]$, add an \flowedge{} $q_{x, i} \to q_x$ of cost $w_i(x)$ and an \flowedge{} $q_x \to q_{x, i}$ of cost $0$.
\end{itemize}
\end{definition}

To insert $rv$, we run Dijkstra's algorithm from $q_v$ and find a shortest path to a vertex $q_{b,c}$ where $b$ is free of color $c$.
Along a simple shortest path from $q_v$ to $q_{b, c}$, each \flowedge{} $q_{x, i} \to q_{y, i}$ applies a color-shift of color $i$ at $x$ (moving the uncolored edge to $xy$), each subpath $q_{x, i} \to q_x \to q_{x, j}$ changes the color at $x$ from $i$ to $j$ at cost $w_i(x)$, and reaching $q_{b, c}$ assigns the final uncolored edge color $c$ directed out of $b$.
When $H$ can be partitioned into $k$ forests, the input path supplied by \Cref{lem:star-sequence-exists} can be followed in $Q$, ensuring that some $q_{b,c}$ with $b$ free of color $c$ is reachable from $q_v$.
The following lemma shows that any simple directed path in $Q$ from $q_v$ to such a vertex $q_{b,c}$ yields a complete sequence of color-shifts.

\begin{lemma}[Paths and color-shifts]
\label{obs:star-path-sequence}
Let $P$ be a simple directed path in $Q$ from $q_v$ to $q_{b,c}$, where $b\in V(H)\setminus\{r\}$ is free of color $c$.
Listing $(x,i)$ for each shift \flowedge{} $q_{x,i}\to q_{y,i}$ along $P$ in order, followed by $(b,c)$, forms a complete sequence of color-shifts for $rv$.
\end{lemma}
\begin{proof}
Start with $rv$ uncolored at $v$ and follow $P$ from $q_v$ to $q_{b,c}$:
\begin{itemize}
\item \emph{Color changes via $q_x$:} Passing through $q_x$ chooses the next color at $x$ without changing the current uncolored edge, and simplicity ensures $q_x$ is visited at most once.
\item \emph{Color-shifts on $q_{x,i}\to q_{y,i}$:} The original edge $xy$ initially has color $i$ and direction $x\to y$.
It cannot have been removed earlier: that would require an earlier traversal of this \flowedge{} of $Q$, whereas a simple path visits $q_{x,i}$ only once.
Thus $xy$ is still the outgoing edge of color $i$ at $x$, and we can shift its color onto the current uncolored edge at $x$, leaving $xy$ uncolored at $y$.
\item \emph{Termination at $q_{b,c}$:} At the endpoint $q_{b,c}$, the vertex $b$ is still free of color $c$ by \Cref{obs:star-shift-orientation}.
Since $b\neq r$, coloring the remaining uncolored edge with $c$ directed out of $b$ completes the sequence as in \Cref{def:star-shift-sequence}.
\end{itemize}
\end{proof}

The path gives a complete sequence by \Cref{obs:star-path-sequence}; \Cref{lem:one-edge-dijkstra}, stated below, shows that choosing a shortest path also ensures that the coloring obtained after applying the corresponding complete sequence of color-shifts partitions $E(H)$ into $k$ forests.

\begin{lemma}[Correctness of one-edge insertion]
\label{lem:one-edge-dijkstra}
Suppose $E(H)\setminus\{rv\}$ is partitioned into $k$ forests $F_1,\ldots,F_k$ oriented with a valid orientation with respect to $r$.
For any vertex $b\in V(H)\setminus\{r\}$ free of color $c$, any shortest path in $Q$ from $q_v$ to $q_{b, c}$ gives a complete sequence of color-shifts for $rv$ whose final coloring partitions $E(H)$ into $k$ forests.
\end{lemma}
\begin{proof}
Choose a simple shortest path from $q_v$ to $q_{b,c}$.
By \Cref{obs:star-path-sequence}, it describes a complete sequence of color-shifts, and hence \Cref{obs:star-shift-sequence} gives a coloring of all edges with a valid orientation.

It remains to show that no color contains a directed cycle.
With outdegree at most one, every undirected cycle is directed: its edges have distinct tails, and thus every vertex has one outgoing edge on the cycle.
For each vertex $u$ of $Q$, let $\pi(u)$ be the shortest-path distance from $q_v$ to $u$ in $Q$.
Suppose for contradiction that color $i$ contains a directed cycle $C$.
Since every vertex has outdegree at most one in color $i$ and $r$ is a sink, $C$ is a directed cycle in $H - r$.
At least one edge of $C$ changed color, since $F_i$ was initially acyclic.
For each edge $x\to y$ of $C$, we will prove $\pi(q_{y,i})\leq\pi(q_{x,i})$, with strict inequality when $xy$ changed color.
\begin{itemize}
\item If $xy$ kept color $i$ and direction $x \to y$, then $x \to y$ belonged to $F_i$, and the \flowedge{} $q_{x, i} \to q_{y, i}$ of cost $0$ gives $\pi(q_{y, i}) \leq \pi(q_{x, i})$.
\item If $xy$ kept color $i$ but reversed direction from $y \to x$ to $x \to y$, the shortest path traversed $q_{y, i} \to q_{x, i}$ of cost $0$, giving $\pi(q_{y, i}) = \pi(q_{x, i})$.
\item If $xy$ changed from color $j \neq i$ to color $i$, its initial direction in $F_j$ was $y \to x$, and the shortest path traversed
\begin{equation*}
  q_{y, j} \xrightarrow{0} q_{x, j} \xrightarrow{w_j(x)} q_x \xrightarrow{0} q_{x, i},
\end{equation*}
which yields $\pi(q_{x, i}) = \pi(q_{y, j}) + w_j(x)$.
On the other hand, extending a shortest path from $q_v$ to $q_{y,j}$ by $q_{y, j} \xrightarrow{w_j(y)} q_y \xrightarrow{0} q_{y, i}$ of cost $w_j(y)$ gives $\pi(q_{y, i}) \leq \pi(q_{y, j}) + w_j(y)$.
By \Cref{def:star-auxiliary-graph}, the input edge $y\to x$ of color $j$ gives $w_j(y)=w_j(x)-1$, yielding
\begin{equation*}
  \pi(q_{y, i}) \leq \pi(q_{y, j}) + w_j(y) = \pi(q_{x, i}) - 1 < \pi(q_{x, i}).
\end{equation*}
\end{itemize}
Since at least one edge $x \to y$ on $C$ changed color, $\pi(q_{x, i})$ is finite, and summing these inequalities around $C$ gives $\pi(q_{x, i}) < \pi(q_{x, i})$, a contradiction.
\end{proof}

\subsection{The general algorithm via flow}
\label{sec:star-flow}

The previous subsection handles a single star edge $rv$ by finding a shortest path in the auxiliary graph $Q$ from $q_v$ to a vertex $q_{b, c}$ where $b$ is free of color $c$.
We now want to insert all star edges in $U$ at once.
We start with all star edges in $U$ uncolored and the given forests $F_1, \ldots, F_k$ partitioning $E(H - r)$.
As in \Cref{sec:star-design}, root each component of each $F_i$ arbitrarily and direct edges toward the roots; the depths $d_i(x)$ and costs $w_i(x) = 2m - d_i(x) > 0$ (since $d_i(x) \leq m$) refer to these input forests and remain fixed.

To insert all edges in $U$ at once, we seek one sequence of color-shifts starting at each star edge $rv\in U$.
However, two sequences cannot both assign color $i$ to edges directed out of the same vertex $x$---whether by applying a color-shift of color $i$ at $x$ or by terminating in color $i$ at a vertex $x$ free of color $i$.
We therefore build a capacitated flow network $D$ directly from the auxiliary graph $Q$ by adding a source $s$, a sink $t$, and \flowedge{} capacities so that each unit of $s$-$t$ flow traces one sequence of color-shifts and unit capacities prevent conflicts at any vertex $q_{x, i}$.
We first maximize the flow value (the number of inserted star edges) and then minimize the total cost, which rules out directed cycles in each color just as shortest paths did in \Cref{lem:one-edge-dijkstra}.

The precise construction is given in \Cref{def:star-flow-network} below.

\begin{definition}[Flow network]
\label{def:star-flow-network}
\normalfont
The flow network $D$ is obtained from the auxiliary graph $Q$ by retaining all vertices $q_x$ and $q_{x,i}$, for $x \in V(H) \setminus \{r\}$ and $i \in [k]$, as well as all \flowedge{}s of $Q$, and adding a source $s$ and a sink $t$.
We then assign \flowedge{}s, capacities, and costs as follows:
\begin{enumerate}
\item \emph{Source \flowedge{}s (starting sequences):} For each star edge $rv \in U$, add $s \to q_v$ with capacity $1$ and cost $0$.
Capacity $1$ ensures that $rv$ starts at most one sequence of color-shifts; from $q_v$, a unit of flow chooses its initial color $i$ for $rv$ via the \flowedge{} $q_v \to q_{v, i}$ of $Q$ (see Item~4) at cost $0$.
\item \emph{Color-shift \flowedge{}s:} For each original edge $x \to y$ of color $i$ in $F_i$, keep the \flowedge{} $q_{x, i} \to q_{y, i}$ of $Q$ with cost $0$ and give it capacity $1$.
Traversing $q_{x, i} \to q_{y, i}$ applies a color-shift of color $i$ at $x$: it assigns the current uncolored edge at $x$ color $i$ directed out of $x$, and makes $xy$ the uncolored edge at $y$.
\item \emph{Sink \flowedge{}s (finishing sequences):} For each color $i \in [k]$ and each vertex $x\neq r$ that is free of color $i$, add $q_{x, i} \to t$ with capacity $1$ and cost $0$.
Traversing $q_{x, i} \to t$ finishes a sequence by assigning its current uncolored edge at $x$ color $i$ directed out of $x$.
\item \emph{Color-change \flowedge{}s via $q_x$:} For each $x \neq r$ and $i \in [k]$, keep the \flowedge{}s $q_{x, i} \to q_x$ (cost $w_i(x)$) and $q_x \to q_{x, i}$ (cost $0$) of $Q$, and give both capacity $m + 1$.
A subpath $q_{x, i} \to q_x \to q_{x, j}$ changes our choice of color at $x$ from $i$ to $j$ at cost $w_i(x)$ without yet coloring the uncolored edge at $x$; that edge is colored when the path follows the next \flowedge{} from $q_{x, j}$ to another $q_{y, j}$ or to $t$.
Since the total capacity leaving $s$ is $|U| \leq m$, capacity $m + 1$ exceeds the maximum $s$-$t$ flow value, so these \flowedge{}s never restrict the flow.
\end{enumerate}
\end{definition}

Notice that for every $x \neq r$ and $i \in [k]$, $D$ contains a unique \flowedge{} from $q_{x, i}$ to another $q_{y, i}$ (if $x$ is not free of color $i$) or to $t$ (if $x$ is free of color $i$), corresponding to assigning an outgoing edge of color $i$ at $x$.
Its capacity of $1$ guarantees that at most one sequence assigns an outgoing edge of color $i$ at $x$.
This does not prevent paths from sharing network vertices: a path may pass through $q_{x,i}$ and change color without coloring an edge with color $i$ at $x$.

Let $f(D)$ be the maximum $s$-$t$ flow value in $D$.
Since each source \flowedge{} $s \to q_v$ has capacity $1$, $f(D) \leq |U|$.
If $f(D) < |U|$, the algorithm returns \textsc{No}; the following lemma shows that in this case $H$ has no partition into $k$ forests.

\begin{lemma}[Infeasible case: $f(D) < |U|$]
\label{lem:star-failure}
Let $D$ be the flow network of \Cref{def:star-flow-network} with maximum $s$-$t$ flow value $f(D)$.
If $f(D) < |U|$, then $H$ has no partition into $k$ forests.
\end{lemma}
\begin{proof}
Suppose $f(D) < |U|$.
By the max-flow min-cut theorem, let $(S, T)$ be a minimum $s$-$t$ cut in $D$ with $s \in S$ and $t \in T$, and let $\mathrm{cap}(S, T)$ be the total capacity of \flowedge{}s from $S$ to $T$.
Since $\mathrm{cap}(S, T) = f(D) < |U| \leq m$, no \flowedge{} of capacity $m + 1$ can leave $S$.
If $q_x$ and $q_{x,i}$ were on opposite sides, one of the two \flowedge{}s between them would leave $S$.
Because $q_{x, i} \to q_x$ and $q_x \to q_{x, i}$ both have capacity $m + 1$, for each $x \in V(H) \setminus \{r\}$ the vertex $q_x$ and all $q_{x, i}$ ($i \in [k]$) lie on the same side of $(S, T)$.
Define
\begin{equation*}
  X = \{r\} \cup \{x \in V(H) \setminus \{r\} : q_x \in S\}.
\end{equation*}
We bound $\mathrm{cap}(S,T)$ from below by counting two disjoint groups of \flowedge{}s from $S$ to $T$, each of capacity $1$:
\begin{itemize}
\item \emph{Source \flowedge{}s $s \to q_v$ ($rv \in U$):} Since $s \in S$, the \flowedge{} $s \to q_v$ leaves $S$ unless $q_v \in S$, i.e., $v \in X \setminus \{r\}$ (which means $rv \in U \cap E_H(X)$).
Thus $|U| - |U \cap E_H(X)|$ source \flowedge{}s leave $S$.
\item \emph{\expandafter\MakeUppercase\flowedge{}s leaving $q_{x, i}$ ($x \in X \setminus \{r\}$, $i \in [k]$):} For each of the $k(|X| - 1)$ pairs $(x, i) \in (X \setminus \{r\}) \times [k]$, we have $q_{x, i} \in S$.
If $x$ is free of color $i$ in $F_i$, the capacity-$1$ \flowedge{} $q_{x, i} \to t$ leaves $S$ (as $t \in T$).
Otherwise, $x$ has a unique outgoing edge $x \to y$ in $F_i$ ($y \neq r$), and the capacity-$1$ \flowedge{} $q_{x, i} \to q_{y, i}$ leaves $S$ unless $q_{y, i} \in S$, i.e., $y \in X \setminus \{r\}$ (which gives a distinct edge $xy \in E(F_i) \cap E_H(X)$).
Each such edge determines its tail $x$ and color $i$, so different pairs give different edges.
Thus at most $|E(H - r) \cap E_H(X)|$ of these $k(|X| - 1)$ \flowedge{}s stay inside $S$, and at least $k(|X| - 1) - |E(H - r) \cap E_H(X)|$ of them leave $S$.
\end{itemize}
Combining these bounds gives
\begin{align*}
  |U| &> f(D) = \mathrm{cap}(S, T) && \text{(by max-flow min-cut)}\\
      &\geq \bigl(|U| - |U \cap E_H(X)|\bigr) + \bigl(k(|X| - 1) - |E(H - r) \cap E_H(X)|\bigr) && \text{(by the two bounds above)}\\
      &= |U| + k(|X| - 1) - \bigl(|U \cap E_H(X)| + |E(H - r) \cap E_H(X)|\bigr)\\
      &= |U| + k(|X| - 1) - |E(H) \cap E_H(X)| && \text{(as $E(H) = U \dotcup E(H - r)$)}\\
      &= |U| + k(|X| - 1) - |E_H(X)|. && \text{(as $E_H(X) \subseteq E(H)$)}
\end{align*}
Rearranging gives $|E_H(X)| \geq k(|X| - 1) + 1$ (and hence $|X| \geq 2$, since $|E_H(\{r\})| = 0$).
Because a forest contains at most $|X|-1$ edges inside $X$, $H$ cannot be partitioned into $k$ forests.
\end{proof}

Now that we know the feasible case requires $f(D) = |U|$, suppose $f(D) = |U|$.
Among all flows of this value, choose an integral flow $\varphi$ of minimum cost, writing $\varphi(u, w)$ for the flow on $u \to w$.
The following lemma shows that the \flowedge{}s carrying positive flow under $\varphi$ contain no directed cycle.

\begin{lemma}[No directed cycles in the flow]
\label{obs:star-flow-acyclic}
Let $\varphi$ be an integral minimum-cost maximum $s$-$t$ flow in the flow network $D$ of \Cref{def:star-flow-network}.
The \flowedge{}s carrying positive flow under $\varphi$ contain no directed cycle.
\end{lemma}
\begin{proof}[Proof]
Indeed, any directed cycle in $D$ avoids $s$ and $t$ (as $s$ has indegree $0$ and $t$ has outdegree $0$) and hence lies in $Q$: if it visits some vertex $q_x$, it must enter $q_x$ via an \flowedge{} $q_{x, i} \to q_x$ of cost $w_i(x) > 0$; otherwise, all of its \flowedge{}s have the form $q_{x, i} \to q_{y, i}$ for a single color $i$, which is impossible because $F_i$ is a forest.
Thus every directed cycle in $D$ has strictly positive cost, so no directed cycle can carry positive flow on all of its \flowedge{}s (otherwise, reducing $\varphi$ by $1$ unit around that cycle would preserve the flow value $f(D)$ while strictly decreasing the cost).
Integrality ensures that each used \flowedge{} carries at least one unit, so this reduction preserves nonnegativity, capacities, and flow conservation.
\end{proof}

We next decompose $\varphi$ into unit-flow paths from $s$ to $t$, one for each star edge $rv \in U$.
For each star edge $rv \in U$, we start with $s\to q_v$ and follow \flowedge{}s carrying positive flow in the remaining flow until reaching $t$.
Flow conservation prevents the walk from getting stuck before $t$, and \Cref{obs:star-flow-acyclic} prevents it from revisiting a vertex.
This gives a simple path $P$; since the flow is integral, we can subtract one unit of flow along $P$.
The remaining flow is integral, has value one less, and still has no directed cycle.

We initialize $F'_i \gets F_i$ for each $i \in [k]$ and process the extracted paths one by one.
For a path $P$ starting with $s \to q_v$, we start with $e \gets rv$ as the current uncolored edge and follow the \flowedge{}s of $P$ in order:
\begin{itemize}
\item On $q_{x, i} \to q_{y, i}$, apply a color-shift of color $i$ at $x$: remove $xy$ from $F'_i$, add $e$ to $F'_i$ directed out of $x$, and set $e \gets xy$ as the new uncolored edge.
\item On $q_{x, i} \to t$, finish the sequence: add $e$ to $F'_i$ directed out of $x$.
\item On any \flowedge{} entering or leaving $q_x$ ($s \to q_x$, $q_{x, i} \to q_x$, or $q_x \to q_{x, j}$), leave $e$ and $F'_1, \ldots, F'_k$ unchanged (the uncolored edge $e$ remains incident to $x$).
\end{itemize}
Throughout this process, the sets $F'_i$ may temporarily contain directed cycles, but \Cref{lem:star-forest} proves that at the end each $F'_i$ is a forest.
\Cref{alg:star-insertion} summarizes the complete procedure.

\begin{algorithm}[H]
\caption{$\StarInsertion(H,r,F_1,\ldots,F_k)$}
\label{alg:star-insertion}
\KwIn{Simple graph $H$, vertex $r$, and $k\geq 1$ forests $(F_i)_{i=1}^k$ with $E(H-r)=\bigsqcup_{i=1}^k E(F_i)$}
\KwOut{A partition of $E(H)$ into $k$ forests, or \textnormal{\textsc{No}} if no such partition exists}
$U\gets\{rv\in E(H)\}$\cmt{Star edges to insert}
Root and orient $(F_i)$, and construct the auxiliary graph $Q$ and the flow network $D$ from $Q$ as in \Cref{def:star-flow-network}\;
Compute an integral minimum-cost maximum flow $\varphi$ of value $f(D)$ in $D$\;
\lIf{$f(D)<|U|$}{\Return \textnormal{\textsc{No}}\hfill\Comment*[f]{\textcolor{blue}{No $k$-forest partition exists}}}
$(F'_1,\ldots,F'_k)\gets(F_1,\ldots,F_k)$\;
Decompose $\varphi$ into $|U|$ unit $s$-$t$ paths\cmt{One path per star edge}
\ForEach{path $P$ in the decomposition, starting with $s\to q_v$}{
  $e\gets rv$\cmt{Current uncolored edge}
  \ForEach{network \flowedge{} on $P$}{
    \uIf{the \flowedge{} is $q_{x,i}\to q_{y,i}$}{
      Remove $xy$ from $F'_i$, and add $e$ to $F'_i$ directed out of $x$\cmt{Color-shift of color $i$ at $x$}
      $e\gets xy$\;
    }
    \ElseIf{the \flowedge{} is $q_{x,i}\to t$}{
      Add $e$ to $F'_i$, directed out of $x$\cmt{Finish the sequence at $x$}
    }
  }
}
\Return $(F'_1,\ldots,F'_k)$\;
\end{algorithm}

\Cref{lem:star-paths-orientation}, stated below, extends \Cref{obs:star-path-sequence} across all paths, using the unit capacities to show that applying the paths gives an edge partition with a valid orientation.

\begin{lemma}[Applying the flow paths]
\label{lem:star-paths-orientation}
Let $\varphi$ be an integral minimum-cost maximum $s$-$t$ flow of value $f(D) = |U|$ in the flow network $D$ of \Cref{def:star-flow-network}.
Each path $P$ in the decomposition of $\varphi$ in \Cref{alg:star-insertion} defines a complete sequence of color-shifts (\Cref{def:star-shift-sequence}), and the returned sets $F'_1,\ldots,F'_k$ partition $E(H)$ with a valid orientation with respect to $r$.
Each original edge is either unchanged or recolored once with its direction reversed, and each star edge is colored once toward $r$.
\end{lemma}
\begin{proof}[Proof]
Let $P_1,\ldots,P_{|U|}$ be the $|U|$ paths from $s$ to $t$ extracted from $\varphi$, which are simple by \Cref{obs:star-flow-acyclic}.
Each carries one unit of flow and is applied individually by the loop of \Cref{alg:star-insertion}.
Since their total flow is $|U|$, all $|U|$ source \flowedge{}s of capacity one are used, so every star edge starts a path.
We first verify that $F'_1, \ldots, F'_k$ partition $E(H)$ and that every vertex has outdegree at most $1$ in each $F'_i$:
\begin{itemize}
\item \emph{No original edge is shifted twice.}
Each path $P_\ell$ starts at a distinct star edge $rv \in U$ (since $s \to q_v$ has capacity $1$), shifts colors along a sequence of original edges $x \to y \in F_i$ via \flowedge{}s $q_{x, i} \to q_{y, i}$, optionally changes colors via $q_{x, j} \to q_x \to q_{x, i}$, and finishes at a vertex free of color $i$ via $q_{x, i} \to t$.
Because every color-shift \flowedge{} $q_{x, i} \to q_{y, i}$ has capacity $1$, no original edge is shifted by two different paths.
Thus when $P_\ell$ traverses $q_{x, i} \to q_{y, i}$, the edge $x \to y$ is still in $F'_i$, is removed from $F'_i$, and is recolored once directed out of $y$ ($y \to x$) when $P_\ell$ next follows an \flowedge{} out of $q_{y, i}$ (or $q_{y, j}$ after a single color change through $q_y$), and no subsequent path ever removes it again; meanwhile, every original edge $x \to y \in F_i$ with $\varphi(q_{x, i}, q_{y, i}) = 0$ remains untouched in $F'_i$ directed $x \to y$.
Likewise, each star edge $rv \in U$ is colored once directed $v \to r$ when its path first follows an \flowedge{} out of $q_{v, i}$.
Hence $F'_1, \ldots, F'_k$ partition $E(H)$.
\item \emph{Valid orientation via \Cref{obs:star-shift-sequence}.}
If a path ends with $q_{x,i}\to t$, no other path uses this \flowedge{} because its capacity is $1$.
Earlier shifts preserve free colors by \Cref{obs:star-shift-orientation}, and only a path ending with $q_{x,i}\to t$ can use the free color $i$ at $x$.
Thus $x$ is still free of color $i$ when this path ends, so each path $P_\ell$ forms a complete sequence of color-shifts (\Cref{def:star-shift-sequence}) for its star edge $rv$.
By \Cref{obs:star-shift-sequence}, applying $P_1,\ldots,P_{|U|}$ in order colors all edges in $U$, keeps all original edges colored, and preserves a valid orientation with respect to $r$.
\end{itemize}
\end{proof}

It remains to rule out cycles in $F'_1, \ldots, F'_k$.
We use the minimum cost of $\varphi$ for this in \Cref{lem:star-forest}, stated below.

\begin{lemma}[Feasible case: $f(D) = |U|$]
\label{lem:star-forest}
Let $D$ be the flow network of \Cref{def:star-flow-network} with maximum $s$-$t$ flow value $f(D)$.
If $f(D) = |U|$, then the sets $F'_1, \ldots, F'_k$ returned by \Cref{alg:star-insertion} form a partition of $E(H)$ into $k$ forests.
\end{lemma}
\begin{proof}
By \Cref{lem:star-paths-orientation}, the returned sets partition $E(H)$ and their directions form a valid orientation with respect to $r$.

It remains to show that each $F'_i$ is acyclic.
Suppose for contradiction that some $F'_i$ contains a cycle $C$.
Since every vertex in $F'_i$ has outdegree at most $1$ and $r$ is a sink, $C$ is a simple directed cycle in $H - r$.
Because $F_i$ was a forest, at least one edge of $C$ changed color.

We construct a new integral flow $\varphi'$ of the same value $f(D) = |U|$ and strictly smaller cost by modifying $\varphi$ along each directed edge $x \to y$ of $C$ as follows (see \Cref{fig:flow-cycle-reduction}):
\begin{itemize}
\item \emph{Case 1: $xy$ kept color $i$ and direction $x \to y$.}
Then $x \to y \in F_i$ and $\varphi(q_{x, i}, q_{y, i}) = 0$ (otherwise $xy$ would have been reversed to $y \to x$).
In $\varphi'$, we add $1$ unit of flow to $q_{x, i} \to q_{y, i}$ (cost change $0$).
\item \emph{Case 2: $xy$ kept color $i$ but reversed direction from $y \to x$ to $x \to y$.}
Then $y \to x \in F_i$ and a flow path traversed $q_{y, i} \to q_{x, i}$ to recolor $xy$ out of $x$, so $\varphi(q_{y, i}, q_{x, i}) = 1$.
In $\varphi'$, we remove $1$ unit of flow from $q_{y, i} \to q_{x, i}$ (cost change $0$).
\item \emph{Case 3: $xy$ changed from color $j \neq i$ to color $i$.}
Because $xy$ changed color, it was shifted by $\varphi$, and every shifted original edge reverses direction when recolored at its head; thus its input direction in $F_j$ was $y \to x$, and a flow path traversed
\begin{equation*}
  q_{y, j} \xrightarrow{0} q_{x, j} \xrightarrow{w_j(x)} q_x \xrightarrow{0} q_{x, i}
\end{equation*}
at cost $w_j(x)$.
Instead of shifting $y \to x$ in color $j$ and changing to color $i$ at $x$, we change from color $j$ to color $i$ at $y$: in $\varphi'$, we remove $1$ unit of flow from $q_{y, j} \to q_{x, j} \to q_x \to q_{x, i}$ and add $1$ unit of flow to
\begin{equation*}
  q_{y, j} \xrightarrow{w_j(y)} q_y \xrightarrow{0} q_{y, i}.
\end{equation*}
Since $y$ is a child of $x$ in $F_j$, we have $d_j(y) = d_j(x) + 1$ and $w_j(y) = w_j(x) - 1$, so this change lowers the flow cost by $w_j(x) - w_j(y) = 1$.
\end{itemize}
In all three cases, the modification for $x \to y$ increases the net flow leaving $q_{x, i}$ by $1$, decreases the net flow leaving $q_{y, i}$ by $1$, and preserves flow conservation at every other vertex of $D$.
At each vertex $x$ of $C$, the modification for its outgoing edge adds $1$ to the net flow leaving $q_{x,i}$, while the modification for its incoming edge subtracts $1$.
These changes cancel, so $\varphi'$ satisfies flow conservation and has the same value as $\varphi$.
No source or sink \flowedge{} is modified.
Moreover, since $C$ is a simple cycle in $H - r$, its edges $xy$ are distinct (so each unit-capacity color-shift \flowedge{} changes at most once, between $0$ and $1$), and each vertex $u \in V(C)$ serves once as a tail $x$ (decreasing flow on $q_{u, j} \to q_u \to q_{u, i}$ by at most $1$) and once as a head $y$ (increasing flow on $q_{u, j'} \to q_u \to q_{u, i}$ by at most $1$).
On $q_u \to q_{u, i}$ (and on $q_{u, j} \to q_u$ if $j = j'$), if both changes occur they have opposite signs and cancel; since every \flowedge{} incident to $q_u$ carries at most $f(D) \leq m$ units under $\varphi$ against a capacity of $m + 1$, the flow on every \flowedge{} remains nonnegative and does not exceed its capacity.
Finally, since at least one edge of $C$ changed color (Case 3), the total cost of $\varphi'$ is strictly smaller than that of $\varphi$, contradicting the optimality of $\varphi$.
Hence each $F'_i$ is a forest.
\end{proof}

\usetikzlibrary{backgrounds}
\begingroup
\definecolor{StarRed}{HTML}{B64040}
\definecolor{StarBlue}{HTML}{265FA8}
\definecolor{FlowGreen}{HTML}{1B7837}
\definecolor{FlowOrange}{HTML}{D95F02}
\setlength{\intextsep}{7pt}
\begin{figure}[H]
\centering
\tikzset{
  fcvertex/.style={paper graph vertex,minimum size=3.4mm,font=\scriptsize},
  fcroot/.style={fcvertex,double,double distance=0.45pt,minimum size=3.7mm,font=\scriptsize\bfseries},
  fcgray/.style={paper pending edge,line width=0.55pt,dash pattern=on 2pt off 1.5pt},
  fcred/.style={paper forest arc,line width=0.68pt,-{Latex[length=1.35mm,width=0.92mm]},draw=StarRed},
  fcblue/.style={paper forest arc,line width=0.68pt,-{Latex[length=1.35mm,width=0.92mm]},draw=StarBlue},
  fcnewred/.style={paper highlighted arc,line width=1.65pt,-{Latex[length=1.85mm,width=1.32mm]},draw=StarRed},
  fcnewblue/.style={paper highlighted arc,line width=1.65pt,-{Latex[length=1.85mm,width=1.32mm]},draw=StarBlue},
  fchub/.style={circle,draw=black!55,fill=white,line width=0.3pt,minimum size=3.4mm,inner sep=0.25pt,font=\scriptsize},
  fcport/.style={rounded corners=0.7pt,draw=black!55,fill=white,line width=0.3pt,minimum width=4.8mm,minimum height=3.3mm,inner sep=0.3pt,font=\scriptsize},
  fcedge/.style={paper graph arc,line width=0.45pt,-{Latex[length=1.05mm,width=0.72mm]},draw=black!48},
  fcflowA/.style={paper highlighted arc,line width=1.45pt,-{Latex[length=1.68mm,width=1.18mm]},draw=FlowGreen},
  fcflowB/.style={paper highlighted arc,line width=1.45pt,-{Latex[length=1.68mm,width=1.18mm]},draw=FlowOrange},
  fccost/.style={fill=white,inner sep=0.25pt,font=\fontsize{5.5}{6.5}\selectfont}
}
\begin{subfigure}[b]{0.15\textwidth}
\centering
\begin{tikzpicture}[x=0.72cm,y=0.78cm]
  \path (0,-1.25) (1.1,1.68);
  \node[fcroot]   (r) at (0.55, 1.0) {$r$};
  \node[fcvertex] (b) at (0.0,  0.0) {$b$};
  \node[fcvertex] (a) at (1.1,  0.0) {$a$};
  \node[fcvertex] (c) at (0.55,-1.0) {$c$};

  \begin{scope}[on background layer]
    \draw[fcgray] (r) -- (a);
    \draw[fcgray] (r) -- (b);
    \draw[fcred]  (b) -- (a);
    \draw[fcblue] (b) -- (c);
    \draw[fcblue] (c) -- (a);
  \end{scope}
\end{tikzpicture}
\subcaption{Input $H$.}
\label{fig:flow-cycle-input}
\end{subfigure}%
\hfill%
\begin{subfigure}[b]{0.40\textwidth}
\centering
\begin{tikzpicture}[x=0.72cm,y=0.78cm]
  \path (0,-1.25) (7.3,1.68);
  \node[fchub]  (s)   at (0.0,  0.0) {$s$};
  \node[fcport] (qb1) at (1.15, 1.0) {$q_{b,1}$};
  \node[fchub]  (qb)  at (1.15, 0.0) {$q_b$};
  \node[fcport] (qb2) at (1.15,-1.0) {$q_{b,2}$};
  \node[fcport] (qc1) at (2.7,  1.0) {$q_{c,1}$};
  \node[fchub]  (qc)  at (2.7,  0.0) {$q_c$};
  \node[fcport] (qc2) at (2.7, -1.0) {$q_{c,2}$};
  \node[fcport] (qa1) at (4.25, 1.0) {$q_{a,1}$};
  \node[fchub]  (qa)  at (4.25, 0.0) {$q_a$};
  \node[fcport] (qa2) at (4.25,-1.0) {$q_{a,2}$};
  \node[fchub]  (t)   at (5.3,  0.0) {$t$};

  \node[fcroot]   (r) at (6.75, 1.0) {$r$};
  \node[fcvertex] (b) at (6.2,  0.0) {$b$};
  \node[fcvertex] (a) at (7.3,  0.0) {$a$};
  \node[fcvertex] (c) at (6.75,-1.0) {$c$};

  \begin{scope}[on background layer]
    \draw[fcedge] (qb1) to[bend right=18] node[pos=0.5,fccost,left=1.1pt] {$9$} (qb);
    \draw[fcedge] (qb)  to[bend left=18] (qb2);
    \draw[fcedge] (qb2) to[bend left=18] node[pos=0.23,fccost,left=1.1pt] {$8$} (qb);
    \draw[fcedge] (qc)  to[bend right=18] (qc1);
    \draw[fcedge] (qc1) to[bend right=18] node[pos=0.5,fccost,left=1.1pt] {$10$} (qc);
    \draw[fcedge] (qc)  to[bend left=18] (qc2);
    \draw[fcedge] (qc2) to[bend left=18] node[pos=0.23,fccost,left=1.1pt] {$9$} (qc);
    \draw[fcedge] (qa2) to[bend left=18] node[pos=0.23,fccost,left=1.1pt] {$10$} (qa);
    \draw[fcedge] (qb2) -- (qc2);
    \draw[fcedge] (qc2) -- (qa2);
    \draw[fcedge] (qc1.north) .. controls (3.55,1.78) and (4.85,1.62) .. (t.north);

    \draw[fcflowA] (s.south east) .. controls (1.4,-0.46) and (2.95,-0.46) .. (qa.south west);
    \draw[fcflowA] (qa)  to[bend right=18] (qa1);
    \draw[fcflowA] (qa1) -- (t);

    \draw[fcflowB] (s)   -- (qb);
    \draw[fcflowB] (qb)  to[bend right=18] (qb1);
    \draw[fcflowB] (qb1.north east) .. controls (2.1,1.62) and (3.3,1.62) .. (qa1.north west);
    \draw[fcflowB] (qa1) to[bend right=18] node[pos=0.5,fccost,left=1.1pt] {$10$} (qa);
    \draw[fcflowB] (qa)  to[bend left=18] (qa2);
    \draw[fcflowB] (qa2) -- (t);

    \draw[fcnewred]  (a) -- (r);
    \draw[fcnewred]  (b) -- (r);
    \draw[fcnewblue] (a) -- (b);
    \draw[fcblue]    (b) -- (c);
    \draw[fcblue]    (c) -- (a);
  \end{scope}
\end{tikzpicture}
\subcaption{Feasible flow (cost $10$, blue cycle).}
\label{fig:flow-cycle-bad}
\end{subfigure}%
\hfill%
\begin{subfigure}[b]{0.40\textwidth}
\centering
\begin{tikzpicture}[x=0.72cm,y=0.78cm]
  \path (0,-1.25) (7.3,1.68);
  \node[fchub]  (s)   at (0.0,  0.0) {$s$};
  \node[fcport] (qb1) at (1.15, 1.0) {$q_{b,1}$};
  \node[fchub]  (qb)  at (1.15, 0.0) {$q_b$};
  \node[fcport] (qb2) at (1.15,-1.0) {$q_{b,2}$};
  \node[fcport] (qc1) at (2.7,  1.0) {$q_{c,1}$};
  \node[fchub]  (qc)  at (2.7,  0.0) {$q_c$};
  \node[fcport] (qc2) at (2.7, -1.0) {$q_{c,2}$};
  \node[fcport] (qa1) at (4.25, 1.0) {$q_{a,1}$};
  \node[fchub]  (qa)  at (4.25, 0.0) {$q_a$};
  \node[fcport] (qa2) at (4.25,-1.0) {$q_{a,2}$};
  \node[fchub]  (t)   at (5.3,  0.0) {$t$};

  \node[fcroot]   (r) at (6.75, 1.0) {$r$};
  \node[fcvertex] (b) at (6.2,  0.0) {$b$};
  \node[fcvertex] (a) at (7.3,  0.0) {$a$};
  \node[fcvertex] (c) at (6.75,-1.0) {$c$};

  \begin{scope}[on background layer]
    \draw[fcedge] (qb1.north east) .. controls (2.1,1.62) and (3.3,1.62) .. (qa1.north west);
    \draw[fcedge] (qb2) to[bend left=18] node[pos=0.23,fccost,left=1.1pt] {$8$} (qb);
    \draw[fcedge] (qc)  to[bend right=18] (qc1);
    \draw[fcedge] (qc1) to[bend right=18] node[pos=0.5,fccost,left=1.1pt] {$10$} (qc);
    \draw[fcedge] (qc)  to[bend left=18] (qc2);
    \draw[fcedge] (qc2) to[bend left=18] node[pos=0.23,fccost,left=1.1pt] {$9$} (qc);
    \draw[fcedge] (qa1) to[bend right=18] node[pos=0.5,fccost,left=1.1pt] {$10$} (qa);
    \draw[fcedge] (qa)  to[bend left=18] (qa2);
    \draw[fcedge] (qa2) to[bend left=18] node[pos=0.23,fccost,left=1.1pt] {$10$} (qa);
    \draw[fcedge] (qc1.north) .. controls (3.55,1.78) and (4.85,1.62) .. (t.north);

    \draw[fcflowA] (s.south east) .. controls (1.4,-0.46) and (2.95,-0.46) .. (qa.south west);
    \draw[fcflowA] (qa)  to[bend right=18] (qa1);
    \draw[fcflowA] (qa1) -- (t);

    \draw[fcflowB] (s)   -- (qb);
    \draw[fcflowB] (qb)  to[bend right=18] (qb1);
    \draw[fcflowB] (qb1) to[bend right=18] node[pos=0.5,fccost,left=1.1pt] {$9$} (qb);
    \draw[fcflowB] (qb)  to[bend left=18] (qb2);
    \draw[fcflowB] (qb2) -- (qc2);
    \draw[fcflowB] (qc2) -- (qa2);
    \draw[fcflowB] (qa2) -- (t);

    \draw[fcnewred]  (a) -- (r);
    \draw[fcred]     (b) -- (a);
    \draw[fcnewblue] (b) -- (r);
    \draw[fcnewblue] (c) -- (b);
    \draw[fcnewblue] (a) -- (c);
  \end{scope}
\end{tikzpicture}
\subcaption{Cheaper flow (cost $9<10$, forests).}
\label{fig:flow-cycle-good}
\end{subfigure}
\caption{Cost reduction along a directed cycle $C=(a\to b\to c\to a)$ in color $2$ (blue) for $m=5$: switching $q_{b,1}\to q_{a,1}\to q_a\to q_{a,2}$ to $q_{b,1}\to q_b\to q_{b,2}$ (Case~3 on $a\to b$) and adding flow on $q_{b,2}\to q_{c,2}\to q_{a,2}$ (Case~1 on $b\to c, c\to a$) lowers the flow cost from $w_1(a)=10$ to $w_1(b)=9$ (the two unit $s$-$t$ flows are shown in green and orange).}
\label{fig:flow-cycle-reduction}
\end{figure}
\endgroup

\begin{lemma}[Correctness of $\StarInsertion$]
\label{cor:star-correctness}
Given a simple graph $H$, a vertex $r \in V(H)$, and a partition of $E(H - r)$ into $k$ forests $F_1, \ldots, F_k$, \Cref{alg:star-insertion} returns a partition of $E(H)$ into $k$ forests if one exists, and returns \textnormal{\textsc{No}} otherwise.
\end{lemma}
\begin{proof}
Since each source \flowedge{} $s \to q_v$ ($rv \in U$) has capacity $1$, the maximum flow value satisfies $f(D) \leq |U|$.
If $f(D) = |U|$, \Cref{alg:star-insertion} returns a partition of $E(H)$ into $k$ forests by \Cref{lem:star-forest}; if $f(D) < |U|$, \Cref{alg:star-insertion} returns \textsc{No}, and no partition of $E(H)$ into $k$ forests exists by \Cref{lem:star-failure}.
\end{proof}

\begin{lemma}[Size and running time]
\label{lem:star-size}
The flow network $D$ of \Cref{def:star-flow-network} has $O(nk)$ vertices and \flowedge{}s, integer capacities at most $m + 1$, and nonnegative integer costs at most $2m$.
Constructing $D$ takes $O(nk + m)$ time, and given an integral minimum-cost maximum flow $\varphi$, \Cref{alg:star-insertion} produces its output in $O(nk + m)$ additional time.
\end{lemma}
\begin{proof}
The network $D$ has $n - 1$ vertices $q_x$, $(n - 1)k$ vertices $q_{x, i}$, and the two terminals $s, t$, for $(n - 1)(k + 1) + 2 = O(nk)$ vertices in total.
Each $q_{x, i}$ has one outgoing \flowedge{} to another $q_{y, i}$ or to $t$, plus the pair of \flowedge{}s $q_{x, i} \to q_x$ and $q_x \to q_{x, i}$, and $s$ has $|U| \leq n - 1$ outgoing \flowedge{}s $s \to q_v$; this gives $3(n - 1)k + |U| = O(nk)$ \flowedge{}s.
Rooting $F_1, \ldots, F_k$, computing the depths $d_i(x)$, and constructing $Q$ and $D$ take $O(nk + m)$ time.

Given $\varphi$, if $f(D) < |U|$, the algorithm returns \textsc{No} in $O(1)$ time.
If $f(D) = |U|$, by \Cref{lem:star-paths-orientation} each edge of $H$ is colored at most once by the $f(D)$ flow paths, so the paths use at most $m$ \flowedge{}s of the form $q_{x, i} \to q_{y, i}$ or $q_{x, i} \to t$ in total.
Between two consecutive colorings (or before the first), a path can pass through a vertex $q_x$ at most once by \Cref{obs:star-flow-acyclic}, using at most two network \flowedge{}s ($q_{x, i} \to q_x \to q_{x, j}$ or $s \to q_v \to q_{v, i}$).
Thus the total length of all $f(D)$ paths is $\sum_{\ell=1}^{f(D)} |P_\ell| \leq 3m = O(m)$.
Initializing $(F'_1, \ldots, F'_k) \gets (F_1, \ldots, F_k)$ and scanning $D$ to list, for each vertex, the outgoing \flowedge{}s carrying positive flow take $O(nk + m)$ time; by always advancing along the first positive-flow \flowedge{} at each vertex and removing it once its remaining flow reaches zero, we trace each unit $s$-$t$ path $P_\ell$ and apply its color-shifts in $O(|P_\ell|)$ time, for $\sum_{\ell=1}^{f(D)} O(|P_\ell|) = O(m)$ additional time and $O(nk + m)$ time in total.
\end{proof}

\subsection[Proof of Theorem~\ref*{thm:star}]{Proof of \Cref{thm:star}}

We complete the proof by including the time needed to compute the minimum-cost maximum flow.

\begin{proof}[Proof of \Cref{thm:star}]
By \Cref{cor:star-correctness}, \Cref{alg:star-insertion} correctly returns a partition of $E(H)$ into $k$ forests if one exists and \textsc{No} otherwise.
To compute the integral minimum-cost maximum flow $\varphi$, we first compute the maximum $s$-$t$ flow value $f(D)$, and then we compute an integral minimum-cost flow with source and sink demands $\pm f(D)$ (which exists since all capacities and demands are integral).
By \Cref{lem:star-size}, $D$ has $O(nk)$ vertices and \flowedge{}s, capacities at most $m + 1$, and costs at most $2m$.
The deterministic exact maximum flow and minimum-cost flow algorithms of \cite[Theorem~1.1]{van2023deterministic} perform both computations in $(nk + m)^{1+o(1)}$ time, and the remaining $O(nk + m)$ work of \Cref{lem:star-size} fits within the same bound.
\end{proof}


\section{Two approaches to exact forest partitioning}
\label{sec:exact}

Let $G$ be a simple graph with $n$ vertices and $m$ edges,
and let $k=\rho(G)$ be its pseudoarboricity. We first compute $k$
and a $k$-out orientation in $m^{1+o(1)}$ time using
\Cref{prop:prelim-minimum-orientation}. We use this orientation throughout the section. Since $k\leq a(G)\leq k+1$, the arboricity of $G$ is either $k$ or $k+1$.
Our goal in this section is to determine whether $a(G)=k$
and, if so, construct a partition into $k$ forests. If $a(G)=k+1$,
the additive-one algorithm provides the required optimal partition,
as explained in \Cref{sec:full-exact}.

The additive-one construction in \Cref{sec:additivewrapper}
already provides a partition into $k+1$ forests in
$O(m\log n+nk\log^2 n)$ time. It therefore remains to construct a
$k$-forest partition or determine that none exists. In the latter
case, $a(G)=k+1$, so the partition obtained above is optimal.

We find a set $R\subseteq V(G)$ of
$O\!\left(\frac{n}{k}\log n\right)$ vertices such that the restricted orientation
of $G-R$ is
a $(k-1)$-out orientation
(\Cref{lem:deletion}). Applying
\Cref{thm:additive} to $G-R$ then yields $k$ forests. Using \StarInsertion\ from \Cref{thm:star}, we insert the vertices of $R$ one at a time, adding each vertex together with its incident edges to the current graph as a star. Each call either extends the partition,
possibly recoloring old edges, or reports that $k$ forests
are impossible. If all calls succeed, we obtain a partition of $G$.
This approach takes
\begin{equation*}
  (m+nk)^{1+o(1)}\cdot\frac{n}{k}
\end{equation*}
time. Since $m\leq nk$ and $k\leq n$, this gives an $n^{2+o(1)}$ running time, which is already almost linear when
$m=\Theta(n^2)$.

For small $k$, we instead use the \emph{rooted construction}
of \Cref{thm:small}. It applies Gabow's round-robin
algorithm~\cite{gabow1991matroid} and solves the
partition problem in $O(nk^2\log n)$ time.
Choosing between the two approaches after the reduction in
\Cref{sec:full-algorithm} gives the final bound
\begin{equation*}
  \min\{n^{2+o(1)},m^{4/3+o(1)}\}.
\end{equation*}

\subsection{First approach: \textsc{DeleteAndReinsert}}
\label{sec:deletion}

Let $G$ have $n$ vertices, $m$ edges,
and a supplied $k$-out orientation $\mathcal O$.
Call a remaining vertex \emph{saturated} if its current outdegree
is $k$. Once no saturated vertices remain, every outdegree is
at most $k-1$. The restricted orientation is therefore
a $(k-1)$-out orientation, so the additive-one algorithm
(\Cref{thm:additive})
gives $k$ forests.

We repeatedly delete a vertex whose deletion makes the most
saturated vertices unsaturated, until no saturated vertices remain. This gives a small deletion set
$R$. After partitioning $G-R$ into $k$ forests, we reinsert
$R$ using $\StarInsertion$. If every \StarInsertion\ call succeeds,
we obtain $k$ forests on $G$, and a failure shows that there is no $k$-forest partition.
\Cref{fig:delete-and-reinsert} illustrates these three steps with an example graph.

A saturated vertex becomes unsaturated when it or one of its out-neighbors is deleted. Let $\ell_v$ be the number of saturated vertices among $v$ and its in-neighbors. Then deleting $v$ makes exactly $\ell_v$ vertices unsaturated. Indeed, a saturated $v$ is removed from the graph and hence no longer counted, while each saturated in-neighbor of $v$ loses the out-neighbor $v$, reducing its outdegree from $k$ to $k-1$. \Cref{alg:deletion-set} stores $\ell_v$ for every remaining vertex $v$ in a maximum heap. When a vertex $u$ becomes unsaturated, we decrease $\ell_v$ by one for every $v\in\{u\}\cup N^+_{\mathcal O}(u)$. Each vertex becomes unsaturated at most once, so there are $O(m+n)$ heap updates, and the procedure takes $O((m+n)\log n)$ time.

\begin{algorithm}[H]
\caption{\textsc{DeletionSet}$(G,k,\mathcal O)$}
\label{alg:deletion-set}
\KwIn{Simple graph $G$ and a $k$-out orientation $\mathcal O$}
\KwOut{A set $R$ such that the restricted orientation of $G-R$ is a $(k-1)$-out orientation}
$R\gets\varnothing$\cmt{Selected deletions}
For every $v$, set $\mathsf{sat}[v]\gets(d^+_{\mathcal O}(v)=k)$\nllabel{line:deletion-init}\;
Set $s$ to the total number of saturated vertices and $\ell_v$ to their number in $\{v\}\cup N^-_{\mathcal O}(v)$\;
Build a maximum heap $B$ on $V(G)$ with keys $\ell_v$\cmt{Largest effect first}
\While{$s>0$}{
  Extract a vertex $r$ with maximum key from $B$; $R\gets R\cup\{r\}$\nllabel{line:deletion-select}\cmt{Greedy deletion}
  \ForEach{$u\in\{r\}\cup N^-_{\mathcal O}(r)$ with $\mathsf{sat}[u]=\mathsf{true}$}{
    $\mathsf{sat}[u]\gets\mathsf{false}$, $s\gets s-1$\nllabel{line:deletion-clear}\cmt{No longer saturated}
    \ForEach(\Comment*[f]{\textcolor{blue}{Update affected $\ell_v$ values}}){$v\in N^+_{\mathcal O}(u)\cup\{u\}$}{
      $\ell_v\gets\ell_v-1$\nllabel{line:deletion-count}\cmt{Stop counting $u$}
      \If{$v\notin R$}{Decrease the key of $v$ in $B$ to $\ell_v$\nllabel{line:deletion-key}\cmt{Update its priority}}
    }
  }
}
\Return $R$\cmt{The remaining orientation is a $(k-1)$-out orientation}
\end{algorithm}

\begin{lemma}[A small deletion set]
\label{lem:deletion}
Let $G$ have $n$ vertices, $m$ edges, and a $k$-out
orientation $\mathcal O$. \Cref{alg:deletion-set} finds a set of vertices $R$ with
\begin{equation*}
  |R|=O\left(\frac{n}{k}\log n\right)
\end{equation*}
in $O((m+n)\log n)$ time so that
the restricted orientation of $G-R$ is a $(k-1)$-out orientation.
\end{lemma}
\begin{proof}
We use \Cref{alg:deletion-set}. Line~\ref{line:deletion-init}
initializes the saturation flags. Outdegrees only decrease, so
no new vertex becomes saturated. Deleting $r$ makes exactly the
saturated vertices in $\{r\}\cup N^-_{\mathcal O}(r)$ unsaturated;
line~\ref{line:deletion-clear} updates their flags and $s$.
Each such vertex $u$ stops contributing to $\ell_v$ for
$v\in N^+_{\mathcal O}(u)\cup\{u\}$. The updates in lines~\ref{line:deletion-count}
and~\ref{line:deletion-key} therefore keep the values $\ell_v$ and heap
keys correct. Thus line~\ref{line:deletion-select} always
chooses a vertex whose deletion makes the most saturated vertices
unsaturated.
When $s=0$, the restricted orientation is $(k-1)$-out.

Write $s^{(i)}$ and $\ell_v^{(i)}$ for the values at
the start of iteration $i$. Every saturated vertex $u$ still
has all $k$ of its original out-neighbors. Thus $N^+_{\mathcal O}(u)\cup\{u\}$ has $k+1$
remaining vertices, and $u$ contributes one to $\ell_v^{(i)}$ for each $v\in N^+_{\mathcal O}(u)\cup\{u\}$. Hence
\begin{equation*}
  \sum_{v\notin R}\ell_v^{(i)}=(k+1)s^{(i)}.
\end{equation*}
Since $r$ maximizes
$\ell_v^{(i)}$ over the remaining vertices, and there are at most
$n$ such vertices, the maximum is at least the sum divided by $n$:
\begin{equation*}
  \ell_r^{(i)}\geq
  \frac{\sum_{v\notin R}\ell_v^{(i)}}{n}
  =\frac{(k+1)s^{(i)}}{n}.
\end{equation*}
Deleting $r$ makes exactly $\ell_r^{(i)}$ saturated vertices
unsaturated. Thus, at each iteration, at least a
$\frac{k+1}{n}$ fraction of the saturated vertices become
unsaturated, leaving at most a $1-\frac{k+1}{n}$ fraction.
Applying this bound $t$ times, starting with $s^{(0)}\leq n$,
and using $1-x\leq e^{-x}$ gives
\begin{equation*}
  s^{(t)}\leq n\left(1-\frac{k+1}{n}\right)^t
  \leq n\exp\left(-\frac{(k+1)t}{n}\right).
\end{equation*}
After $O\left(\frac{n}{k}\log n\right)$ iterations,
this bound is below one, so $s=0$. This bounds $|R|$.

Finally, constructing the neighbor lists, flags, values $\ell_v$, and
heap takes $O(m+n)$ time. Each vertex $u$ is cleared at most
once in line~\ref{line:deletion-clear}, causing $|N^+_{\mathcal O}(u)\cup\{u\}|$ updates to the values $\ell_v$ in line~\ref{line:deletion-count}. For the initially
saturated vertices, these lists
have total size at most $m+n$. The incoming lists scanned at
selected vertices have total size at most $m$, since each
vertex is selected at most once. There are therefore
$O(m+n)$ heap operations, each taking $O(\log n)$ time.
This gives the claimed $O((m+n)\log n)$ bound.
\end{proof}

We can now combine deletion (\Cref{alg:deletion-set}), the additive-one construction (\Cref{thm:additive}), and \StarInsertion\ calls (\Cref{thm:star}) into a single algorithm, \textsc{DeleteAndReinsert} (\Cref{alg:delete-reinsert}), which constructs a $k$-forest partition of $G$ or reports that none exists. \Cref{fig:delete-and-reinsert} illustrates these steps with an example.

\begin{algorithm}[H]
\caption{\textsc{DeleteAndReinsert}$(G,k,\mathcal O)$}
\label{alg:delete-reinsert}
\KwIn{Graph $G$, $k\geq 2$, and a $k$-out orientation $\mathcal O$}
\KwOut{A $k$-forest partition, or \textnormal{\textsc{No}} if none exists}
$R\gets\textsc{DeletionSet}(G,k,\mathcal O)$ from \Cref{alg:deletion-set}\cmt{Make the orientation of $G-R$ a $(k-1)$-out orientation}
$W\gets V(G)\setminus R$\cmt{Vertices currently present}
$\mathcal F\gets\,$\AdditiveOneCall{$G[W],k-1,\mathcal{O}[W]$}\nllabel{line:delete-reinsert-color}\cmt{Convert to $k$ forests}
\ForEach{$r\in R$ }{
  Apply $\StarInsertion(G[W\cup\{r\}],r,\mathcal F)$ from \Cref{alg:star-insertion}\nllabel{line:delete-reinsert-star}\cmt{Reinsert $r$}
  \If{\textnormal{\textsc{No}} is returned}{\Return \textnormal{\textsc{No}}\cmt{Report impossibility}}
  Replace $\mathcal F$ by the returned partition and set $W\gets W\cup\{r\}$\;
}
\Return $\mathcal F$\cmt{A $k$-forest partition of $G$}
\end{algorithm}

\begin{figure}[htbp]
\centering
\begingroup
\tikzset{
  delete reinsert vertex/.style={paper graph vertex},
  delete reinsert saturated/.style={draw=violet,fill=violet!8},
  delete reinsert removed/.style={draw=orange!90!black,fill=orange!10},
  delete reinsert arc/.style={paper graph arc},
  delete reinsert blue/.style={paper forest edge,draw=blue!75!black},
  delete reinsert green/.style={paper forest edge,draw=ForestGreen},
  delete reinsert step/.style={paper step},
  delete reinsert procedure/.style={above=3pt,inner sep=1.5pt,
    font=\footnotesize}
}
\newcommand{\DeleteReinsertSquare}[2]{%
  \node[delete reinsert vertex,#2] (#1a) at (0,0.65) {$a$};
  \node[delete reinsert vertex,#2] (#1b) at (1.15,0.65) {$b$};
  \node[delete reinsert vertex] (#1c) at (0,-0.65) {$c$};
  \node[delete reinsert vertex,#2] (#1d) at (1.15,-0.65) {$d$};
}
\newcommand{\DeleteReinsertVertices}[1]{%
  \node[delete reinsert vertex,delete reinsert removed]
    (#1r) at (-1.15,0.65) {$r$};
  \node[delete reinsert vertex,delete reinsert removed]
    (#1s) at (-1.15,-0.65) {$s$};
}
\resizebox{\linewidth}{!}{%
\begin{tikzpicture}[x=0.96cm,font=\small,label distance=1pt]
  \begin{scope}
    \DeleteReinsertSquare{q}{delete reinsert saturated}
    \DeleteReinsertVertices{q}
    \draw[delete reinsert arc] (qa) -- (qb);
    \draw[delete reinsert arc] (qb) -- (qd);
    \draw[delete reinsert arc] (qd) -- (qc);
    \draw[delete reinsert arc] (qc) -- (qa);
    \draw[delete reinsert arc] (qa) -- (qr);
    \draw[delete reinsert arc] (qb)
      .. controls (1,1.55) and (-1,1.55) .. (qr);
    \draw[delete reinsert arc] (qr) -- (qs);
    \draw[delete reinsert arc] (qs) -- (qc);
    \draw[delete reinsert arc] (qd)
      .. controls (1,-1.55) and (-1,-1.55) .. (qs);
    \node at (0,1.8) {$2$-out};
    \node at (0,-1.8) {$G$};
  \end{scope}
  \begin{scope}[xshift=4.1cm]
    \DeleteReinsertSquare{w}{}
    \draw[delete reinsert arc] (wa) -- (wb);
    \draw[delete reinsert arc] (wb) -- (wd);
    \draw[delete reinsert arc] (wd) -- (wc);
    \draw[delete reinsert arc] (wc) -- (wa);
    \node at (0.575,1.8) {$1$-out};
    \node at (0.575,-1.8) {$G-R$};
  \end{scope}
  \begin{scope}[xshift=8.5cm]
    \DeleteReinsertSquare{f}{}
    \draw[delete reinsert blue] (fa) -- (fb);
    \draw[delete reinsert blue] (fb) -- (fd);
    \draw[delete reinsert blue] (fd) -- (fc);
    \draw[delete reinsert green] (fc) -- (fa);
    \node at (0.575,1.8) {$2$ forests};
    \node at (0.575,-1.8) {$\mathcal F$};
  \end{scope}
  \begin{scope}[xshift=13.5cm]
    \DeleteReinsertSquare{g}{}
    \DeleteReinsertVertices{g}
    \draw[delete reinsert blue] (ga) -- (gb);
    \draw[delete reinsert blue] (gb) -- (gd);
    \draw[delete reinsert blue] (gd) -- (gc);
    \draw[delete reinsert blue] (gc) -- (gs);
    \draw[delete reinsert blue] (gs) -- (gr);
    \draw[delete reinsert green] (gc) -- (ga);
    \draw[delete reinsert green] (ga) -- (gr);
    \draw[delete reinsert green] (gr)
      .. controls (-1,1.55) and (1,1.55) .. (gb);
    \draw[delete reinsert green] (gd)
      .. controls (1,-1.55) and (-1,-1.55) .. (gs);
    \node at (0,1.8) {$2$ forests};
    \node at (0,-1.8) {$\mathcal F'$};
  \end{scope}
  \draw[delete reinsert step] (1.65,-0.15) --
    node[delete reinsert procedure] {Delete $R$} (3.8,-0.15);
  \draw[delete reinsert step] (5.9,-0.15) --
    node[delete reinsert procedure] {\AdditiveOne} (8.45,-0.15);
  \draw[delete reinsert step] (10.35,-0.15) --
    node[delete reinsert procedure] {\StarInsertion} (12.6,-0.15);
\end{tikzpicture}%
}
\endgroup
\caption{\textsc{DeleteAndReinsert} for $k=2$ on a six-vertex
graph $G$. The orange vertices form $R=\{r,s\}$, and the violet
vertices are saturated (outdegree $2$).
Deleting $R$ leaves a $1$-out orientation. \mbox{\AdditiveOne}\ converts
$G-R$ into the blue and green forests $\mathcal F$.
\mbox{\StarInsertion}\ reinserts $r,s$ and their incident edges,
giving the two forests $\mathcal F'$.}
\label{fig:delete-and-reinsert}
\end{figure}

\thmDeleteAndReinsert*
\begin{proof}
By \Cref{lem:deletion}, the orientation of $G-R$ is
a $(k-1)$-out orientation. Thus line~\ref{line:delete-reinsert-color}
constructs $k$ forests by \Cref{thm:additive}. At each
\StarInsertion\ call, line~\ref{line:delete-reinsert-star} either extends
the partition to $G[W\cup\{r\}]$ or correctly reports that no such partition exists, by \Cref{thm:star}.
In the latter case, $G$ has no $k$-forest partition either, since
any such partition would restrict to a partition of $G[W \cup \{r\}]$. If every \StarInsertion\ call succeeds,
all deleted vertices and their incident edges have been reinserted,
so the final forests partition all edges of $G$.

Deletion takes $O((m+n)\log n)$ time by \Cref{lem:deletion}, and
the initial partition takes $O(m\log n+nk\log^2 n)$ time by
\Cref{thm:additive}. Each \StarInsertion\ call takes $(m+nk)^{1+o(1)}$ time
by \Cref{thm:star}, and there are $O\left(\frac{n}{k}\log n\right)$
vertices to reinsert by \Cref{lem:deletion}. So the total is
$(m+nk)^{1+o(1)}\cdot\frac{n}{k}$ since the logarithmic factors are absorbed by the $(m+nk)^{o(1)}$ term.
\end{proof}

\subsection{Second approach: Rooted construction}
\label{sec:small}

For small $k$, we use a variant of the round-robin algorithm of
\citet{gabow1991matroid} to obtain an $O(nk^2\log n)$ running time,
improving on the bound from \textsc{DeleteAndReinsert} when $k^2<n$.
We later balance the two bounds in \Cref{sec:full-algorithm} to obtain the final running time.
To apply Gabow's algorithm, we reduce partitioning $E(G)$ into
$k$ forests to partitioning the arcs of a directed multigraph into
$k$ spanning trees. This reduction has been used previously
(see, e.g., \cite{berczi2010packing,arkhipov2026faster}). We include the detailed reduction here for completeness.
We construct this multigraph from $G$ by adding
a root vertex and arcs to it, as follows.

Let $G=(V,E)$ be a simple graph with $n$ vertices, $m$ edges, and a
supplied $k$-out orientation $\mathcal O$. We define a directed
multigraph $\mathcal O_r$ with vertex set $V\cup\{r\}$, where $r$ is a new
vertex called the \emph{root}. The arcs of $\mathcal O_r$ are the edges of $G$,
directed as in $\mathcal O$, together with $k-d^+_{\mathcal O}(v)$
parallel arcs from $v$ to $r$ for each $v\in V$. This number is
nonnegative since $\mathcal O$ is a $k$-out orientation. Every
vertex of $V$ has outdegree exactly $k$ in $\mathcal O_r$, and $r$ has
outdegree zero. Thus $\mathcal O_r$ has $n+1$ vertices and $nk$ arcs.

Write $A(\mathcal O_r)$ for the arc set of $\mathcal O_r$. For $X\subseteq V(\mathcal O_r)$,
define
\begin{equation*}
  \delta_{\mathcal O_r}^+(X)
  =\bigl|\{uv\in A(\mathcal O_r):u\in X,\ v\notin X\}\bigr|.
\end{equation*}
Thus $\delta_{\mathcal O_r}^+(X)$ is the number of arcs leaving $X$: their
tails lie in $X$ and their heads lie outside $X$. Parallel arcs
are counted separately.

\Cref{fig:rooted-completion} illustrates the construction for
$k=2$ and a cut with $\delta_{\mathcal O_r}^+(X)=3$.
\begin{figure}[H]
\centering
\begingroup
\definecolor{completionblue}{RGB}{35,88,137}
\tikzset{
  completion vertex/.style={paper graph vertex},
  completion root/.style={completion vertex,draw=orange!85!black,
    fill=orange!10},
  completion edge/.style={paper graph edge,draw=black},
  completion arc/.style={paper graph arc,draw=black,
    shorten <=.6pt,shorten >=.6pt},
  completion added/.style={completion arc,draw=orange!85!black,
    dash pattern=on 2.5pt off 1.7pt},
  completion cut/.style={paper forest edge,draw=completionblue},
  completion step/.style={paper step,
    draw=black!70,line width=.8pt}
}
\newcommand{\RootedCompletionVertices}[1]{%
  \node[completion vertex] (#1a) at (0,1) {$a$};
  \node[completion vertex] (#1b) at (1.8,1) {$b$};
  \node[completion vertex] (#1c) at (1.8,-1) {$c$};
  \node[completion vertex] (#1d) at (0,-1) {$d$};
}
\begin{tikzpicture}[x=1cm,y=1cm,font=\small,
  label distance=2pt,every node/.style={text=black}]
  \begin{scope}
    \RootedCompletionVertices{o}
    \draw[completion arc] (oa) -- (ob);
    \draw[completion arc] (oa) -- (oc);
    \draw[completion arc] (ob) -- (oc);
    \draw[completion arc] (od) -- (oa);
    \draw[completion arc] (od) -- (oc);
    \node at (.9,-1.8) {$(G,\mathcal O)$};
  \end{scope}

  \draw[completion step] (2.32,0) --
    node[above=4pt] {$$} (3.98,0);

  \begin{scope}[xshift=4.5cm]
    \filldraw[rounded corners=4pt,fill=completionblue!5,
      draw=completionblue!60,line width=.6pt,
      dash pattern=on 3pt off 2pt] (-.35,.58) rectangle (2.15,1.63);
    \RootedCompletionVertices{d}
    \node[completion root] (dr) at (4,0) {$r$};

    \draw[completion arc] (da) -- (db);
    \draw[completion arc] (dd) -- (da);
    \draw[completion arc] (dd) -- (dc);
    \draw[completion arc,completion cut] (da) -- (dc);
    \draw[completion arc,completion cut] (db) -- (dc);
    \draw[completion added,completion cut] (db) -- (dr);
    \draw[completion added] (dc) to[bend left=13] (dr);
    \draw[completion added] (dc) to[bend right=13] (dr);

    \node[text=completionblue] at (.9,1.38) {$X$};
    \node[text=completionblue] at (3.25,1.35) {$\delta_{\mathcal O_r}^+(X)=3$};
    \node at (2,-1.8) {$\mathcal O_r$};
  \end{scope}
\end{tikzpicture}
\endgroup
\captionsetup{font=small,skip=6pt}
\caption{From $G$ to $\mathcal O_r$ for $k=2$, using the supplied orientation
$\mathcal O$. Dashed arcs are added: $b\to r$ once and $c\to r$
twice. The shaded set is $X=\{a,b\}$. Its three outgoing arcs
are blue, giving $\delta_{\mathcal O_r}^+(X)=3$.}
\label{fig:rooted-completion}
\end{figure}

An \emph{in-arborescence} of $\mathcal O_r$ is a spanning tree of $\mathcal O_r$ in which
every arc is directed towards $r$. Equivalently, every vertex other
than $r$ has exactly one outgoing arc in the tree. We use two
characterizations of Edmonds~\cite{edmonds1973edge,edmonds2003submodular}, as stated by
\citet{gabow1991matroid} with all arcs reversed.
\begin{itemize}
\item \emph{Cut Characterization.} The maximum number of arc-disjoint in-arborescences of $\mathcal O_r$
directed towards $r$ equals the minimum, over all nonempty
$X\subseteq V(\mathcal O_r)\setminus\{r\}$, of the number of arcs leaving
$X$. Equivalently, this maximum equals
\begin{equation*}
  \min_{\varnothing\neq X\subseteq V(\mathcal O_r)\setminus\{r\}}
    \delta_{\mathcal O_r}^+(X).
\end{equation*}
For the multigraph constructed here, $V(\mathcal O_r)\setminus\{r\}=V$.
Thus the sets in this minimum are exactly the nonempty subsets
of $V$, and the maximum is
$\min_{\varnothing\neq X\subseteq V}\delta_{\mathcal O_r}^+(X)$.
\item \emph{Matroid Characterization.} The arcs of $\mathcal O_r$ can be
partitioned into $k$ in-arborescences directed towards $r$ if and only if they can be
partitioned into $k$ spanning trees, ignoring directions, and every
vertex other than $r$ has outdegree $k$.
\end{itemize}
The first characterization relates in-arborescences of $\mathcal O_r$ directed towards $r$ to
forest partitions of $G$, as shown in \Cref{lem:small}.
The second shows that the directions can be ignored, since every
vertex of $V$ has outdegree $k$ in $\mathcal O_r$ by construction. It
therefore suffices to partition the arcs of $\mathcal O_r$ into $k$ spanning
trees. Gabow's round-robin algorithm provides the required construction.
We now state its guarantee explicitly in the notation used throughout
this paper, in terms of directed multigraphs and spanning trees.
All arc directions are reversed relative to Gabow's statement.

\needspace{7\baselineskip}
\begin{proposition}[Gabow~\cite{gabow1991matroid}, Theorem~2.1 and its multigraph extension]
\label{prop:gabow-round-robin}
Let $\mathcal O_r$ be a directed multigraph with $n$ vertices and $m$ arcs,
counting parallel arcs separately. Let $r$ be a vertex of $\mathcal O_r$, and
let $k$ be an integer. There is a deterministic algorithm
that, in $O\left(k(m+kn)\log n\right)$ time, either returns $k$
arc-disjoint spanning trees of $\mathcal O_r$, ignoring directions, whose union
has outdegree zero at $r$ and outdegree $k$ at every other vertex,
or correctly reports that no such trees exist.
\end{proposition}

\begin{lemma}
\label{lem:small}
The edges of $G$ can be partitioned into $k$ forests if and only if
the arcs of $\mathcal O_r$ can be partitioned into $k$ in-arborescences directed towards $r$.
\end{lemma}
\begin{proof}
Suppose the arcs of $\mathcal O_r$ are partitioned into $k$ in-arborescences directed towards $r$.
Deleting $r$ from each of them leaves a forest on $V$. The arcs of
$\mathcal O_r$ not incident with $r$ are exactly the edges of $G$, so every
edge of $G$ lies in exactly one of these forests.

Conversely, suppose the edges of $G$ can be partitioned into $k$
forests, and let $X\subseteq V$ be nonempty (therefore
$r\notin X$). Restricting each forest
to $X$ leaves a forest on $|X|$ vertices, with at most $|X|-1$
edges. Every edge of $G[X]$ belongs to exactly one of these
restricted forests. Summing their edge counts therefore gives
\begin{equation*}
  |E_G(X)|\leq k(|X|-1).
\end{equation*}

By construction, every vertex of $X$ has outdegree exactly $k$
in $\mathcal O_r$: the added arcs to $r$ fill its original outdegree deficit.
Summing these outdegrees gives $k|X|$ arcs whose tails lie in $X$,
each counted once at its tail. Some of these arcs stay inside $X$,
and the rest leave $X$.

The arcs staying inside $X$ correspond exactly to the edges of
$G[X]$. Indeed, each edge of $G[X]$ contributes one directed arc
with both endpoints in $X$, while every added arc ends at
$r\notin X$ and therefore cannot stay inside $X$. Thus exactly
$|E_G(X)|$ of the $k|X|$ counted arcs are internal to $X$.
Subtracting them leaves precisely the arcs with tail in $X$ and
head outside $X$. Their number is
\begin{equation*}
  \delta_{\mathcal O_r}^+(X)=k|X|-|E_G(X)|
  \geq k|X|-k(|X|-1)
  =k.
\end{equation*}

Since $X$ was arbitrary, at least $k$ arcs leave every nonempty
subset of $V$, so
$\min_{\varnothing\neq X\subseteq V}\delta_{\mathcal O_r}^+(X)\geq k$.
The Cut Characterization therefore guarantees
$k$ arc-disjoint in-arborescences of $\mathcal O_r$ directed towards $r$.
Each spans the $n+1$
vertices of $\mathcal O_r$ and hence has $n$ arcs. Their arc-disjointness
means that together they contain $nk$ distinct arcs. Since $\mathcal O_r$
has exactly $nk$ arcs, they use every arc of $\mathcal O_r$ and form the
required partition.
\end{proof}

\begin{theorem}
\label{thm:small}
Let $G$ be a simple graph with $n$ vertices and a supplied $k$-out
orientation. In $O(nk^2\log n)$ time, we can either partition the
edges of $G$ into $k$ forests or correctly report that no such
partition exists.
\end{theorem}
\begin{proof}
We construct $\mathcal O_r$ in $O(nk)$ time and run the algorithm of
\Cref{prop:gabow-round-robin} on $\mathcal O_r$ with root $r$. Since $\mathcal O_r$ has
$n+1$ vertices and $nk$ arcs, this takes
$O\bigl(k(nk+k(n+1))\log(n+1)\bigr)=O(nk^2\log n)$ time.

If the algorithm returns trees, each has $n$ arcs, so together
they use all $nk$ arcs of $\mathcal O_r$. Deleting $r$ from each tree leaves
a forest, and every edge of $G$ lies in exactly one of them. If the
algorithm reports that no such trees exist, the Matroid
Characterization shows that the arcs of $\mathcal O_r$ cannot be partitioned
into $k$ in-arborescences directed towards $r$. \Cref{lem:small} then shows that $G$
has no partition into $k$ forests.
\end{proof}


\section{The full algorithm}
\label{sec:full-exact}

In this section, we first show how peeling low-degree vertices reduces the
problem to a smaller graph, which we call the \emph{core} and denote by $Q$.
On a nonempty core, the number of edges is proportional to
$k$ times the number of vertices. This relation lets us sharpen
the running-time analysis of both approaches and obtain a better
overall bound, while restoring the peeled vertices takes only
linear time. Finally, we combine the approaches introduced in
\Cref{sec:deletion,sec:small} to present the full algorithm and
prove \Cref{thm:main}.

\subsection{Peeling low-degree vertices}
\label{sec:exactpeeling}

We repeatedly peel vertices of current degree at most $k$ and
call the remaining graph the \emph{core}. It is either empty or
has minimum degree at least $k+1$. The following lemma shows
that each peeled vertex is easy to restore to a partition into
$k$ forests; \Cref{fig:low-degree-insertion} illustrates the insertion.

\begin{lemma}[Low-degree insertion]
\label{lem:low-degree-insertion}
Let $H=(V,E)$ be a graph, and let
$\mathcal F=(F_1,\ldots,F_k)$ be a partition of $E(H)$ into
$k\geq 1$ forests. If we add a new vertex $v$ and $d(v)\leq k$ edges between $v$
and vertices of $H$, the partition extends to the enlarged graph by
inserting the new edges in $O(d(v))$ time.
\end{lemma}
\begin{proof}
Enumerate the new edges as $vu_1,\ldots,vu_{d(v)}$ and insert
$vu_i$ into $F_i$ for each $i\in[d(v)]$. Each forest gains at
most one edge incident with the new vertex $v$. Hence a new cycle
would have to pass through $v$, where it would require two edges
of the same forest, which is impossible. The extended edge sets
therefore form a partition into $k$ forests. Each insertion takes
constant time, giving $O(d(v))$ time in total.
\end{proof}
\begin{figure}[htbp]
\centering
\begin{tikzpicture}[
  vertex/.style={paper graph vertex},
  inserted vertex/.style={vertex,draw=orange!80!black,
    fill=orange!10},
  edge label/.style={fill=white,text=black,inner sep=1.2pt},
  first/.style={paper forest edge,draw=red!75!black},
  second/.style={paper forest edge,draw=blue!75!black},
  third/.style={paper forest edge,draw=ForestGreen},
  pending/.style={paper pending edge},
  every node/.style={font=\small}]
  \newcommand{\LowDegreeVertices}{%
    \node[vertex] (u1) at (0,0) {$u_1$};
    \node[vertex] (u2) at (3,0) {$u_2$};
    \node[vertex] (u3) at (1.5,2.6) {$u_3$};
    \node[inserted vertex] (v) at (1.5,0.9) {$v$};
    \draw[first] (u1) -- node[midway,edge label] {$1$} (u2);
    \draw[second] (u2) -- node[midway,edge label] {$2$} (u3);
    \draw[third] (u3) -- node[midway,edge label] {$3$} (u1);
  }
  \begin{scope}
    \LowDegreeVertices
    \draw[pending] (v) -- (u1);
    \draw[pending] (v) -- (u2);
    \draw[pending] (v) -- (u3);
    \node at (1.5,-0.8) {(a) Before insertion};
  \end{scope}
  \begin{scope}[xshift=7cm]
    \LowDegreeVertices
    \draw[first,paper highlighted edge] (v) -- node[midway,edge label] {$1$} (u1);
    \draw[second,paper highlighted edge] (v) -- node[midway,edge label] {$2$} (u2);
    \draw[third,paper highlighted edge] (v) -- node[midway,edge label] {$3$} (u3);
    \node at (1.5,-0.8) {(b) After insertion};
  \end{scope}
\end{tikzpicture}
\caption{Inserting a vertex of degree three when $k=3$.
Numbers indicate colors, and the dashed edges in (a) are new.
In (b), $vu_i$ receives color $i$, so each forest gains only
one edge at $v$ and remains a forest.}
\label{fig:low-degree-insertion}
\end{figure}

\begin{lemma}[Peeling and reverse restoration]
\label{lem:peeling}
Let $G$ be a graph with $n$ vertices and $m$ edges, let $k\geq 1$ be an
integer, and let $Q$ be the core left after repeatedly peeling vertices of
current degree at most $k$. Then $E(G)$ can be partitioned into $k$ forests
if and only if $E(Q)$ can. Moreover, $Q$ can be computed in $O(n+m)$ time,
and any partition of $E(Q)$ into at least $k$ forests extends to a
partition of $E(G)$ into the same number of forests in $O(n+m)$ time.
\end{lemma}
\begin{proof}
Maintain a queue of vertices of degree at most $k$. When removing
a vertex, record its edges to the remaining vertices and update
their degrees. Each vertex and edge is processed a constant number
of times, so peeling takes $O(n+m)$ time.

In reverse removal order, each restored vertex has at most $k$
edges to vertices already present. By \Cref{lem:low-degree-insertion},
we can insert these edges while preserving the forests.
Every edge removed during peeling is restored exactly once, giving $O(n+m)$ total time.
This extension gives the backward direction of the equivalence.
For the forward direction, restrict a partition of $G$ into $k$ forests to $Q$.
\end{proof}

\Cref{alg:peel} implements this procedure.
Its record $\mathcal L$ lists each removed vertex together with
its edges to vertices still present at the time of removal.

\begin{algorithm}[H]
\caption{\textsc{Peel}$(G,k)$}
\label{alg:peel}
\KwIn{Graph $G=(V,E)$ and an integer $k\geq 1$}
\KwOut{Core $Q$ and peeling record $\mathcal L$}
Mark every vertex active and set $c(v)\gets d_G(v)$ for each $v\in V$\;
Initialize a queue $J$ with all vertices $v$ for which $c(v)\leq k$\;
$\mathcal L\gets$ an empty list\cmt{Record the removal order}
\While{$J$ is not empty}{
  Remove the first vertex $v$ from $J$\cmt{Peel a low-degree vertex}
  $S\gets\{vu\in E:u\text{ is active}\}$\cmt{Edges to restore later}
  Mark $v$ inactive and append $(v,S)$ to $\mathcal L$\;
  \ForEach{edge $vu\in S$}{
    $c(u)\gets c(u)-1$\cmt{Update the remaining degree}
    \If{$c(u)=k$}{Add $u$ to the end of $J$\;}
  }
}
$Q\gets G[\{v\in V:v\text{ is active}\}]$\cmt{The remaining core}
\Return $(Q,\mathcal L)$\;
\end{algorithm}

\begin{lemma}[Size of the core]
\label{lem:core-size}
Let $G$ have a $k$-out
orientation, and let $Q$ be the nonempty core left after repeatedly
peeling vertices of current degree at most $k$.
Writing $N=|V(Q)|$ and $M=|E(Q)|$, we have
\begin{equation*}
  (k+1)N\leq 2M\leq 2kN.
\end{equation*}
In particular, $M=\Theta(kN)$.
\end{lemma}
\begin{proof}
Every vertex of $Q$ has degree at least $k+1$, so summing
degrees gives $2M\geq(k+1)N$. A $k$-out orientation of $G$
restricts to a $k$-out orientation of $Q$, so summing
outdegrees gives $M\leq kN$. Together these inequalities
give $M=\Theta(kN)$.
\end{proof}

By running the algorithm of \Cref{thm:additive} on the core, where $M=\Theta(kN)$ by \Cref{lem:core-size}, we can partition $E(G)$ into $\rho(G)+1$ forests in almost-linear time. Since $\rho(G)\leq a(G)\leq \rho(G)+1$ by \Cref{lem:prelim-density-arboricity}, this partition uses at most $a(G)+1$ forests.

\begin{corollary}[Almost-linear additive-one decomposition]
\label{cor:near-linear-additive-one}
Given a simple graph $G$ with $n$ vertices, $m$ edges, and a $k$-out orientation,
a partition of $E(G)$ into $k+1$ forests can be constructed in
$O(m\log^2 n)$ time. Consequently, given $G$ alone, a partition into at most
$a(G)+1$ forests can be constructed in $m^{1+o(1)}$ time.
\end{corollary}
\begin{proof}
Repeatedly peel vertices of current degree at most $k$,
retaining the supplied orientation on the remaining core. If the core
is nonempty, let it have $N$ vertices and $M$ edges. By
\Cref{lem:core-size}, $Nk=O(M)$. Applying \Cref{thm:additive} to the
restricted orientation therefore constructs $k+1$ forests on the core in
\begin{equation*}
  O(M\log N+Nk\log^2 N)=O(M\log^2 N)=O(m\log^2 n)
\end{equation*}
time. If the core is empty, start with $k+1$ empty forests instead.
By \Cref{lem:low-degree-insertion}, the peeled vertices can be restored
in reverse order in $O(m)$ time. Together with
the linear-time peeling, this gives the claimed bound.

For the second claim, use \Cref{prop:prelim-minimum-orientation} to
compute $k=\rho(G)$ and a corresponding $k$-out orientation in
$m^{1+o(1)}$ time. The construction above then gives $k+1$ forests.
By \Cref{lem:prelim-density-arboricity}, $k\leq a(G)\leq k+1$, so this
partition uses at most $a(G)+1$ forests. The total
running time remains $m^{1+o(1)}$.
\end{proof}

\subsection[Proof of Theorem~\ref*{thm:main}]{Proof of \Cref{thm:main}}
\label{sec:full-algorithm}
\label{sec:bounds}
\label{sec:wrapper}

\begingroup
\setlength{\intextsep}{6pt}
\begin{figure}[H]
\centering
\usetikzlibrary{backgrounds}
\definecolor{arborflowblue}{RGB}{35,88,137}
\definecolor{arborflowgreen}{RGB}{44,100,73}
\tikzset{
  arbor box/.style={rectangle,rounded corners=3pt,
    draw=black!50,fill=white,line width=.7pt,
    align=center,inner xsep=4pt,inner ysep=4pt},
  arbor new/.style={arbor box,draw=arborflowblue,
    fill=arborflowblue!4,line width=.8pt},
  arbor arrow/.style={-{Latex[length=2.2mm,width=1.5mm]},
    draw=black!65,line width=.8pt,rounded corners=3pt},
  arbor line/.style={draw=black!65,line width=.8pt,
    rounded corners=3pt}
}
\begin{tikzpicture}[x=1cm,y=1cm,font=\fontsize{11}{13.2}\selectfont]
  \node[arbor box,text width=4.0cm,minimum height=1.9cm]
    (orient) at (2.30,0)
    {\textbf{Compute $k=\rho(G)$}\\
     and a $k$-out orientation\\
     \Cref{prop:prelim-minimum-orientation}};
  \node[arbor box,text width=4.0cm,minimum height=1.9cm]
    (peel) at (7.90,0)
    {\textbf{Peel vertices}\\
     of degree at most $k$\\
     Core $Q$; \Cref{lem:peeling}};
  \node[arbor new,text width=4.0cm,minimum height=1.9cm]
    (backup) at (13.50,0)
    {\textbf{Build $k+1$ forests}\\
     of $Q$; keep as a backup\\
     \Cref{thm:additive}};
  \draw[arbor arrow] (orient.east) -- (peel.west);
  \draw[arbor arrow] (peel.east) -- (backup.west);

  \node[arbor box,text width=2.6cm,minimum height=1.9cm]
    (choose) at (1.50,-5.30)
    {\textbf{Try $k$ forests}\\
     Choose the\\faster approach};
  \draw[arbor arrow] (backup.south) -- (13.50,-1.70) --
    (1.50,-1.70) -- (choose.north);

  \node[arbor box,fill=black!6,text width=10.85cm,minimum height=1.9cm]
    (known) at (9.825,-3.20)
    {\textbf{Previously known: the $O(kM\log N)$ approach}\\
     Gabow's rooted construction: $O(kM\log N)$ time\\
     \Cref{thm:small}; \cite{gabow1991matroid}};

  \begin{scope}[on background layer]
    \fill[rounded corners=4pt,arborflowblue!3]
      (4.25,-4.60) rectangle (15.40,-8.55);
  \end{scope}
  \node[text=arborflowblue,inner sep=0pt] at (9.825,-4.95)
    {\textbf{This work: \textsc{DeleteAndReinsert}}};
  \node[arbor new,text width=3.0cm,minimum height=2.1cm]
    (delete) at (6.09,-6.45)
    {\textbf{Find a small $R$}\\
     $(k-1)$-out\\
     on $Q-R$\\
     \Cref{lem:deletion}};
  \node[arbor new,text width=3.0cm,minimum height=2.1cm]
    (convert) at (9.82,-6.45)
    {\textbf{Additive-one}\\
     $k$ forests\\
     of $Q-R$\\
     \Cref{thm:additive}};
  \node[arbor new,text width=3.0cm,minimum height=2.1cm]
    (reinsert) at (13.55,-6.45)
    {\textbf{\StarInsertion}\\
     Reinsert $R$\\
     vertex by vertex\\
     \Cref{thm:star}};
  \node[text=arborflowblue,inner sep=0pt] at (9.825,-8.10)
    {$O\!\left(M^{1+o(1)}\cdot\frac{N}{k}\right)$
     time in total; \Cref{thm:deletion-reinsertion}};
  \draw[arbor line] (choose.east) -- (3.55,-5.30);
  \draw[arbor arrow] (3.55,-5.30) -- (3.55,-3.20) -- (known.west);
  \draw[arbor arrow] (3.55,-5.30) -- (3.55,-6.45) -- (delete.west);
  \draw[arbor arrow] (delete.east) -- (convert.west);
  \draw[arbor arrow] (convert.east) -- (reinsert.west);

  \node[arbor box,draw=arborflowgreen,fill=arborflowgreen!4,
      text width=14.9cm,minimum height=1.7cm]
    (finish) at (8.00,-9.75)
    {\textbf{$k$ forests found:} keep them.\qquad
     \textbf{Otherwise:} use the kept $k+1$ forests.\\
     Restore peeled vertices in reverse order and return an optimal partition.\\
     \Cref{lem:peeling,thm:main}};
  \draw[arbor line] (known.east) -- (16.20,-3.20) -- (16.20,-6.45);
  \draw[arbor line] (reinsert.east) -- (16.20,-6.45);
  \draw[arbor arrow] (16.20,-6.45) -- (16.20,-9.75) -- (finish.east);
\end{tikzpicture}
\captionsetup{font=small,skip=6pt}
\caption{Overview of $\Arboricity(G)$, with $N=|V(Q)|$ and $M=|E(Q)|$.
Either approach finds $k$ forests or reports impossibility.}
\label{fig:arboricity-flowchart}
\end{figure}
\endgroup

We now combine the two approaches from \Cref{sec:exact} into $\Arboricity(G)$ (\Cref{alg:arboricity}). 
\Cref{fig:arboricity-flowchart} gives the full overview and
highlights the new components.
First, we compute $k=\rho(G)$ and a $k$-out orientation using
\Cref{prop:prelim-minimum-orientation}, and then we peel vertices
of current degree at most $k$. If the core is empty, reverse restoration immediately gives an optimal partition into $k$
forests. Otherwise, we apply the additive-one algorithm to the
core and save its partition $\mathcal A$ into $k+1$ forests.
By \Cref{lem:prelim-density-arboricity}, $a(G)\in\{k,k+1\}$,
so we then try to partition the core into $k$ forests. We choose
between the rooted construction of \Cref{thm:small} and
\textsc{DeleteAndReinsert} (\Cref{alg:delete-reinsert}) according
to their running-time bounds. A successful attempt gives an optimal
$k$-forest partition. If the chosen
approach reports that no partition into $k$ forests exists, we use
the $(k+1)$-forest partition $\mathcal A$ constructed by the
additive-one algorithm. Finally,
we restore all peeled vertices in reverse order.

\begin{algorithm}[H]
\caption{$\Arboricity(G)$}
\label{alg:arboricity}
\label{alg:wrapper}
\KwIn{Simple undirected graph $G$ with $m$ edges}
\KwOut{$a(G)$ and an optimal forest partition of $E(G)$}
Compute $(k,\mathcal O)$ using \Cref{prop:prelim-minimum-orientation}\nllabel{line:arboricity-orient}\cmt{$k=\rho(G)\leq a(G)$}
$(Q,\mathcal L)\gets\,$\Peel{$G,k$}\nllabel{line:arboricity-peel}\cmt{Peel vertices easy to restore}
$a\gets k$, $\mathcal F\gets$ a tuple of $k$ empty forests\;
\If{$Q\neq\varnothing$}{
  $N\gets |V(Q)|$, $M\gets |E(Q)|$, $\mathcal O\gets\mathcal O[V(Q)]$\cmt{Work only on the core}
  $\mathcal A\gets\,$\AdditiveOneCall{$Q,k,\mathcal O$}\nllabel{line:arboricity-fallback}\cmt{Save the $(k+1)$-forest fallback}
  \eIf{$k^2\leq N$}{
    $\mathcal F\gets$ the rooted construction of \Cref{thm:small} on $(Q,k,\mathcal O)$\nllabel{line:arboricity-rooted}\cmt{$O(kM\log N)$ time}
  }{
    $\mathcal F\gets\textsc{DeleteAndReinsert}(Q,k,\mathcal O)$\nllabel{line:arboricity-delete}\cmt{$M^{1+o(1)}\cdot\frac{N}{k}$ time}
  }
  \If{the chosen approach reports that no partition into $k$ forests exists}{
    $a\gets k+1$, $\mathcal F\gets\mathcal A$\nllabel{line:arboricity-use-fallback}\cmt{$a(G)=k+1$; use the $(k+1)$-forest fallback}
  }
}
Restore $\mathcal L$ in reverse order using \Cref{lem:low-degree-insertion}\nllabel{line:arboricity-restore}\cmt{Restore peeled vertices}
\Return $(a,\mathcal F)$\cmt{Optimal value and a witnessing partition}
\end{algorithm}

We now prove \Cref{thm:main} by analyzing \Cref{alg:arboricity}.

\begin{proof}[Proof of \Cref{thm:main}]

We will first prove the correctness of the algorithm. 
Line~\ref{line:arboricity-orient} computes $k=\rho(G)$, so
$k\leq a(G)\leq k+1$ by \Cref{lem:prelim-density-arboricity}.
By \Cref{lem:peeling}, $G$ has a $k$-forest partition if and only if $Q$
does, and line~\ref{line:arboricity-restore} extends any partition of $Q$
into $k$ or $k+1$ forests to $G$. If $Q$ is empty, the $k$ empty forests
therefore extend to an optimal partition. Otherwise, the chosen approach
either finds $k$ forests of $Q$, which extend to an optimal partition, or
correctly reports that none exist, by \Cref{thm:small} when $k^2\leq N$ and by \Cref{thm:deletion-reinsertion} otherwise. In case the chosen approach reports that no partition into $k$ forests exists, $G$ has no $k$-forest
partition either, so $a(G)=k+1$ and the
$(k+1)$-forest partition $\mathcal A$ from \Cref{thm:additive} extends to an optimal partition.

We now bound the running time. Computing $k$ and its orientation takes
$m^{1+o(1)}$ time by \Cref{prop:prelim-minimum-orientation}, and peeling
and restoring take $O(n+m)$ time by \Cref{lem:peeling}. If $Q$ is empty,
nothing else is done. Otherwise, write $N=|V(Q)|$ and $M=|E(Q)|$, so
$N\leq n$ and $M\leq m$. Since $Nk=O(M)$ by \Cref{lem:core-size},
\Cref{thm:additive} builds $\mathcal A$ in $O(M\log^2 N)$ time.
It remains to bound the time of the attempt to find $k$ forests of $Q$.

The rooted construction takes $O(Nk^2\log N)=O(kM\log N)$ time by
\Cref{thm:small}, and \textsc{DeleteAndReinsert} takes
$M^{1+o(1)}\cdot\frac{N}{k}$ time by \Cref{thm:deletion-reinsertion}.
Up to $m^{o(1)}$ factors, these bounds are $kM$ and $\frac{MN}{k}$, and
they are equal when $k^2=N$. The test $k^2\leq N$ therefore picks the
faster approach, so the attempt takes
$m^{o(1)}\cdot M\min\{k,\frac{N}{k}\}$ time. We bound this minimum in two
ways, using \Cref{lem:core-size} in both. First, $M\leq kN$ gives
\begin{equation*}
  M\min\left\{k,\frac{N}{k}\right\}\leq\frac{MN}{k}\leq N^2\leq n^2.
\end{equation*}
Second, $N\leq\frac{2M}{k}$ gives
\begin{equation*}
  M\min\left\{k,\frac{N}{k}\right\}
  \leq M\min\left\{k,\frac{2M}{k^2}\right\}
  \leq 2M^{4/3},
\end{equation*}
since $k\leq M^{1/3}$ or $\frac{2M}{k^2}<2M^{1/3}$. Since $m\leq n^2$
and $M\leq m$, the total running time is
\begin{equation*}
  m^{1+o(1)}+m^{o(1)}\min\{n^2,m^{4/3}\}
  =\min\{n^{2+o(1)},m^{4/3+o(1)}\}. \qedhere
\end{equation*}
\end{proof}

\begin{corollary}[Almost-linear time for dense graphs]
\label{cor:dense-arboricity}
For a simple graph $G$ with $n$ vertices and $m=\Theta(n^2)$ edges, $\Arboricity(G)$ computes $a(G)$ and a partition of $E(G)$ into $a(G)$ forests in $m^{1+o(1)}$ time.
\end{corollary}


\section*{AI methodology}
\label{sec:ai-methodology}
We used OpenAI's GPT-5.6 Sol to assist with developing the ideas and
proofs in \Cref{sec:additive}, through discussions with the tool, and
GPT-6 Astra in the same way for the results in \Cref{sec:star-insertion}
and the algorithms in \Cref{sec:deletion}. These tools materially
affected \Cref{sec:additive,sec:star-insertion,sec:deletion}. We also
used Codex to improve the clarity of our writing and to draw the
figures. The authors verified the correctness and originality of all
content, including references, and take full responsibility for it.

\clearpage
\printbibliography

\clearpage
\appendix

\section{Proof of the pseudoforest conversion bound}
\label{app:conversion}

We begin with the sets used by the conversion. Let $G$ be the input graph and $P_1,\ldots,P_k$ its pseudoforest partition, with $Z$ cycles in total. For each $i$, remove one edge from each cycle of $P_i$ and put these edges into $M_i$. Set $F_i=P_i\setminus M_i$ and $Y=M_1\dotcup\cdots\dotcup M_k$. Thus each $F_i$ is a forest and $Y$ initially has $Z$ edges. We view $Y$ as a graph on $V(G)$ and call it the surplus graph. Its edges are called surplus edges. An edge has color $i$ if it belongs to $F_i$ or $M_i$.

A partition is valid if each $F_i$ is a forest, each edge of $M_i$ has both endpoints in one tree of $F_i$, and each tree of $F_i$ contains the endpoints of at most one edge of $M_i$. The initial partition has these properties because each component of an input pseudoforest contains at most one cycle. We call the surplus graph colorful if each of its components contains at most one edge of each color. The following lemma, due to \citet[Lemma~10]{BlumenstockFischer}, assumes that the partition is valid and that the surplus graph is colorful. We include its short proof for completeness.

\begin{lemma}[Choice for a cycle exchange]
\label{lem:cycle-exchange-choice}
Let $Y$ be the colorful surplus graph of a valid partition of a simple graph. For every component $C$ of $Y$ containing a cycle and every vertex $v$ of $C$, there is a color $i$ present in $C$ such that no edge of $F_i$ joins $v$ to a vertex of $C$.
\end{lemma}
\begin{proof}
Since $C$ is colorful and contains a cycle, it has at least $|V(C)|$ different colors. Since $G$ is simple, $v$ has at most $|V(C)|-1$ neighbors in $C$. Each edge belongs to at most one forest $F_i$, so at most $|V(C)|-1$ colors have a forest edge joining $v$ to a vertex of $C$. Hence at least one color present in $C$ has no such edge.
\end{proof}

\begin{proof}[Proof of \Cref{prop:conversion}]
The initial sets $F_i$, $M_i$, and $Y$ defined above can be found in $O(m+nk)$ time. We change these sets while preserving a valid partition, until $Y$ becomes a forest. The required output will then be $F_1,\ldots,F_k,Y$.

First, the validity conditions imply that each $M_i$ is a matching. Indeed, two edges of $M_i$ sharing an endpoint would have all their endpoints in the same tree of $F_i$, contradicting the condition that a tree contains at most one such edge. We use this property in the swaps below.

First consider two surplus edges $uv\in M_i$ and $vw\in M_j$. Since each $M_i$ is a matching, we have $i\neq j$. Remove these edges from $M_i$ and $M_j$, and swap their colors. If $v$ and $w$ lie in the same tree of $F_i$, put $vw$ into $M_i$. Otherwise, put it into $F_i$. Do the same for $uv$ with color $j$.

To check the conditions for color $i$, consider the tree of $F_i$ containing $u$ and $v$. After removing $uv$ from $M_i$, this tree has no surplus edge of color $i$. If $w$ lies in the same tree, $vw$ becomes its only surplus edge. Otherwise, $vw$ joins this tree to another tree, which has at most one surplus edge of color $i$. The joined tree therefore also has at most one such edge. The same argument applies to color $j$. Thus the swap preserves all the conditions, and it either keeps both edges in $Y$ or removes at least one of them.

We first use these swaps to make $Y$ colorful, so that each component contains at most one edge of each color. Suppose a component has two edges of color $i$. Take a path whose first and last edges have color $i$, with no other edge of color $i$ between them. Move color $i$ along the path by swapping consecutive edges. At least one edge must leave $Y$ before the two edges of color $i$ become incident, since otherwise $M_i$ would no longer be a matching. Thus a path of length $O(k)$ lets us remove an edge using $O(k)$ swaps.

To find such a path, search a component and record the colors of the edges visited. Stop when a color appears twice. This happens after at most $k+1$ edge visits, and the visited edges contain a path between the two edges of that color. Hence finding the path takes $O(k)$ time. If the search finishes without finding a repeated color, the component is colorful and has at most $k$ edges. Repeating this process makes every component colorful.

We next remove the cycles of $Y$. Let $C$ be a colorful component containing a cycle, and choose an edge $uv$ on a cycle of $C$. By \Cref{lem:cycle-exchange-choice}, there is a color $i$ present in $C$ such that no edge of $F_i$ joins $v$ to a vertex of $C$. We explain how to find this color efficiently below.

Move color $i$ to $uv$ by swapping along a sequence of incident edges of $C$, without repeating an edge. Such a sequence can be found in $O(k)$ time because $C$ has at most $k$ edges. If an edge leaves $Y$, restart the search for cycles. Otherwise, these swaps only change colors, so $uv$ is still on a cycle of $Y$ and the forests have not changed.

In the latter case, $uv\in M_i$, so its endpoints lie in the same tree of $F_i$. Let $vw$ be the first edge on the path from $v$ to $u$ in this tree. Move $uv$ into $F_i$ and $vw$ into $M_i$. Adding $uv$ creates a cycle, and removing $vw$ breaks it without changing the vertices of the tree. Thus the conditions for color $i$ still hold. By the choice of $i$, the vertex $w$ lies outside $C$. Also, removing $uv$ leaves $C$ connected because $uv$ was on a cycle. Hence the new surplus edge $vw$ joins two different components of $Y$ and is a bridge. We call this a cycle exchange.

The joined components have at most $2k$ edges altogether. Use the path swaps again to remove repeated colors from the joined component. Once $Y$ is colorful again, continue looking for cycles. When no cycle remains, the sets $F_1,\ldots,F_k,Y$ give the required forest partition.

We now bound the number of changes. No operation increases the number of surplus edges, so at most $Z$ surplus edges are removed without replacement. Each cycle exchange removes an edge on a cycle and adds a bridge. It creates no new cycle, so the number of edges lying on cycles decreases by at least one. Removing surplus edges cannot increase this number either. Initially there are at most $Z$ such edges, so there are at most $Z$ cycle exchanges. In particular, the algorithm terminates. We call an \emph{attempt} one run of swaps that moves a color along a path, either to remove a repeated color or to move color $i$ to a cycle edge. Each attempt uses $O(k)$ swaps before an edge leaves $Y$ or a cycle exchange takes place. Thus there are $O(Zk)$ swaps in total.

For the running time, keep a union-find structure for each $F_i$. Its elements are the initial trees of $F_i$, and each vertex records its initial tree. Computing these trees and initializing the structures takes $O(m+nk)$ time. Use the standard list implementation with union by size. A find reads the stored representative in constant time. A union changes the representatives of the elements in the smaller set.

There are $O(Zk)$ finds and at most $Z$ unions, since only an insertion between two forest trees joins components. At most $2Z$ elements take part in unions. Each time an element's representative changes, the size of its set at least doubles. Thus each representative changes $O(\log n)$ times, since $Z\leq m<n^2$. The total union-find cost is therefore $O(Zk+Z\log n)$. A cycle exchange leaves the forest components unchanged, so it needs no union-find update.

We also store each $F_i$ in a link-cut tree~\cite{sleator1981data}. This supports inserting or deleting an edge, choosing a root, and finding the parent of a vertex, each in $O(\log n)$ amortized time. To find $vw$ for a cycle exchange, root the tree at $u$ and take $w$ to be the parent of $v$. Then delete $vw$ and insert $uv$. Each cycle exchange therefore uses a constant number of these operations, as does each insertion of a surplus edge into a forest. There are $O(Z)$ such changes, costing $O(Z\log n)$ in total. Creating the trees and inserting their initial edges takes $O(nk+m\log n)$ time.

It remains to find the color $i$ used in a cycle exchange. We can answer the needed edge queries using a fixed orientation of the input graph. In each tree component of each input $P_i$, direct the edges toward a chosen root. In each component containing a cycle, direct the cycle consistently and all remaining edges toward the cycle. Each vertex then has at most one outgoing edge in each $P_i$, and hence at most $k$ outgoing edges in $G$. Keep this orientation fixed throughout the algorithm.

At each vertex, store its outgoing neighbors in a sorted list. Store the current color of each edge and whether it belongs to a forest or to $Y$ with the corresponding list entry. To test whether $vx$ is a forest edge, search the lists at $v$ and $x$. Both lists have length at most $k$, so this takes $O(\log k)$ time. Sorting all lists takes $O(m\log n+n)$ time, since each list also has length at most $n-1$. Once an edge is found, its color and whether it is surplus can be updated in constant time. Keep a pointer to its list entry with each edge, so a swap or cycle exchange does not need to search for that edge again.

To find a color for a cycle exchange, let $C$ be the current colorful cyclic component and let $uv$ be the chosen edge on a cycle of $C$. Mark all colors present in $C$. For each vertex $x$ of $C$ other than $v$, test whether $vx$ is a forest edge. If it is, unmark its color. Any color still marked has no edge in its forest joining $v$ to a vertex of $C$, as required. By \Cref{lem:cycle-exchange-choice}, at least one such color exists. The component has $O(k)$ vertices and edges, so these tests and the search for a marked color take $O(k\log n)$ time. Each attempt leads to an edge leaving $Y$ or to a cycle exchange, so their total cost is $O(Zk\log n)$.

Finally, we must also count searches that finish without removing an edge or exchanging one. A component already known to be colorful need not be searched for repeated colors again until it changes. Likewise, a component already known to be acyclic need not be searched for cycles again until it changes. A completed color search finds at most one edge per color, so it examines at most $k$ edges. A cycle search also examines at most $k$ edges because its component is colorful. Just after a cycle exchange, the joined component has at most $2k$ edges. Initially there are at most $n$ components. Removing an edge creates at most one additional component, and a cycle exchange requires searching again only in the components it joins. Thus there are $O(n+Z)$ completed searches, costing $O(nk+Zk)$ time. A search that finds a repeated color or a cycle leads to an edge leaving $Y$ or to a cycle exchange. There are $O(Z)$ such searches, each taking $O(k)$ time apart from the swaps and edge queries already counted. Reset only the vertex and color marks touched by a search, so restarting it takes time proportional to the work just done.

Combining these costs gives
\begin{equation*}
  O(m\log n+nk+Zk\log n+Z\log n)
  =O(m\log n+nk+Zk\log n).
\end{equation*}
The last step uses $Z\leq m$, which proves the proposition.
\end{proof}

\end{document}